\documentclass[smallcondensed]{svjour3}

\usepackage{pgfplots}
\pgfplotsset{compat=newest}
\usepackage{tikz}
\usetikzlibrary{arrows,matrix,positioning,patterns}
\usepackage[utf8]{inputenc}
\usepackage[english]{babel}
\usepackage[T1]{fontenc}
\usepackage{epsfig}
\usepackage{amsmath, amssymb, amsbsy}
\usepackage{mathdots}
\usepackage[normalem]{ulem}
\usepackage{xspace}
\usepackage[noend]{algpseudocode}
\usepackage[linesnumbered,ruled,vlined,titlenumbered]{algorithm2e}
\usepackage{algorithmicx}
\usepackage{color}
\usepackage{cite}
\usepackage{booktabs}
\usepackage{verbatim}
\usepackage[colorlinks=true,citecolor=blue,linkcolor=blue,urlcolor=blue]{hyperref}
\usepackage{url}
\usepackage{lipsum}
\usepackage{enumitem}
\usepackage{colortbl}
\usepackage{flushend}
\smartqed %
\usepackage[title]{appendix} %
\makeatletter
\newcommand\footnoteref[1]{\protected@xdef\@thefnmark{\ref{#1}}\@footnotemark}
\newcommand{\mylabel}[2]{#2\def\@currentlabel{#2}\label{#1}}
\makeatother

\newenvironment{mymatrix}{\begin{bmatrix}} {\end{bmatrix} }

\def\ve#1{{\mathchoice{\mbox{\boldmath$\displaystyle #1$}}%
              {\mbox{\boldmath$\textstyle #1$}}%
              {\mbox{\boldmath$\scriptstyle #1$}}%
              {\mbox{\boldmath$\scriptscriptstyle #1$}}}}

\definecolor{lightcolor}{rgb}{0.8,0.8,0.8}

\newcommand{\Fq}{\ensuremath{\mathbb{F}_q}}

\DeclareMathOperator{\extsmallfield}{ext}
\DeclareMathOperator{\rank}{rk}

\DeclareMathOperator{\GL}{GL}
\DeclareMathOperator{\GR}{GR}
\DeclareMathOperator{\Ann}{Ann}
\DeclareMathOperator{\ftemp}{f}

\newcommand{\Code}{\mathcal{C}}

\renewcommand{\ker}{\mathcal{K} }

\newcommand{\0}{\ve{0}}
\renewcommand{\S}{\ve{S}}
\renewcommand{\H}{\ve{H}}
\newcommand{\y}{\ve{y}}
\renewcommand{\a}{\ve{a}}
\renewcommand{\b}{\ve{b}}
\newcommand{\e}{\ve{e}}
\renewcommand{\v}{\ve{v}}
\newcommand{\g}{\ve{g}}

\renewcommand{\c}{\ve{c}}
\renewcommand{\e}{\ve{e}}

\newcommand{\G}{\ve{G}}

\newcommand{\K}{\ve{K}}

\newcommand{\E}{\ve{E}}
\newcommand{\C}{\ve{C}}
\newcommand{\Y}{\ve{Y}}
\newcommand{\XX}{\ve{X}}
\newcommand{\U}{\ve{U}}
\newcommand{\LL}{\ve{L}}

\newcommand{\new}[1]{{\color{blue} #1}}
\newcommand{\A}{\ve{A}}
\newcommand{\B}{\ve{B}}

\newcommand{\I}{\ve{I}}

\newcommand{\D}{\ve{D}}

\newcommand{\Z}{\ve{Z}}
\newcommand{\F}{\mathbb{F}}

\newcommand{\NM}{\mathrm{NM}}
\newcommand{\ZZ}{\mathbb{Z}}
\newcommand{\frk}{\mathrm{frk}}
\newcommand{\rk}{\mathrm{rk}}
\newcommand{\T}{\ve{T}}

\newcommand{\Rq}{{R}}
\newcommand{\Rqm}{{S}}

\newcommand{\Fspace}{\mathcal{F}}
\renewcommand{\r}{\ve{r}}
\newcommand{\s}{\ve{s}}
\newcommand{\x}{\ve{x}}
\newcommand{\Sspace}{\mathcal{S}}
\newcommand{\Espace}{\mathcal{E}}
\newcommand{\Aspace}{\mathcal{A}}
\newcommand{\Bspace}{\mathcal{B}}
\newcommand{\Mspace}{\mathcal{M}}
\newcommand{\Nspace}{\mathcal{N}}

\DeclareMathOperator{\rowspace}{rowspace}

\newcommand{\supp}{\mathrm{supp}_\mathrm{R}}

\newcommand{\maxIdeal}{\mathfrak{m}}
\newcommand{\MaxIdeal}{\mathfrak{M}}

\newcommand{\X}{\mathcal{X}}
\newcommand{\FreeModule}[1]{\mathcal{F}(#1)}
\newcommand{\FreeModuleBasis}[1]{F(#1)}

\newcommand{\mingenset}{$\maxIdeal$-shaped basis\xspace}
\newcommand{\mingensets}{$\maxIdeal$-shaped bases\xspace}
\newcommand{\Mingensets}{$\maxIdeal$-Shaped Bases\xspace}

\newcommand{\softO}{\tilde{O}}

\newcommand{\Mmodule}{\mathcal{M}}
\newcommand{\Fcontainsoneproperty}{base-ring property\xspace}
\newcommand\qbin[3]{\left[\begin{matrix} #1 \\ #2 \end{matrix} \right]_{#3}}

\makeatletter
\newcommand{\removelatexerror}{\let\@latex@error\@gobble}
\makeatother

\begin{document}

\title{Low-Rank Parity-Check Codes over Galois Rings}
\titlerunning{Low-Rank Parity-Check Codes over Galois Rings}

\author{Julian Renner \and Alessandro Neri \and Sven Puchinger }

\institute{
J.~Renner  \at Institute for Communications Engineering, Technical University of Munich (TUM), Germany. \email{julian.renner@tum.de}
  \and
 A.~Neri --- \underline{corresponding author} \at Institute for Communications Engineering, Technical University of Munich (TUM), Germany. \email{alessandro.neri@tum.de}
  \and
S.~Puchinger \at Department of Applied Mathematics and Computer Science, Technical University of Denmark (DTU), Denmark. \email{svepu@dtu.dk}
}

\date{Received: date / Accepted: date}

\maketitle

\begin{abstract}
Low-rank parity-check (LRPC) codes are rank-metric codes over finite fields, which have been proposed by Gaborit \emph{et al.}~(2013) for cryptographic applications.
Inspired by a recent adaption of Gabidulin codes to certain finite rings by Kamche \emph{et al.}~(2019), we define and study LRPC codes over Galois rings---a wide class of finite commutative rings.
We give a decoding algorithm similar to Gaborit \emph{et al.}'s decoder, based on simple linear-algebraic operations.
We derive an upper bound on the failure probability of the decoder, which is significantly more involved than in the case of finite fields.
The bound depends only on the rank of an error, i.e., is independent of its free rank.
Further, we analyze the complexity of the decoder.
We obtain that there is a class of LRPC codes over a Galois ring that can decode roughly the same number of errors as a Gabidulin code with the same code parameters, but faster than the currently best decoder for Gabidulin codes. However, the price that one needs to pay is a small failure probability, which we can bound from above.

\end{abstract}

\section{Introduction}

Rank-metric codes are sets of matrices whose distance is measured by the rank of their difference.
Over finite fields, the codes have found various applications in network coding, cryptography, space-time coding, distributed data storage, and digital watermarking.
The first rank-metric codes were introduced in \cite{de78,ga85a, ro91} and are today called Gabidulin codes.
Motivated by cryptographic applications, Gaborit et al.~introduced \emph{low-rank parity-check (LRPC)} in \cite{gaborit2013low,aragon2019low}.
They can be seen as the rank-metric analogs of low-density parity-check codes in the Hamming metric.
LRPC codes have since had a stellar career, as they are already the core component of a second-round submission to the currently running NIST standardization process for post-quantum secure public-key cryptosystems \cite{melchor2020rollo}.
They are suitable in this scenario due to their weak algebraic structure, which prevents efficient structural attacks.
Despite this weak structure, the codes have an efficient decoding algorithm, which in some cases can decode up to the same decoding radius as a Gabidulin code with the same parameters, or even beyond \cite{aragon2019low}.
A drawback is that for random errors of a given rank weight, decoding fails with a small probability.
However, this failure probability can be upper-bounded \cite{gaborit2013low,aragon2019low} and decreases exponentially in the difference between maximal decoding radius and error rank.
The codes have also found applications in powerline communications \cite{yazbek2017LRPCPowerLine} and network coding \cite{8377229}.

Codes over finite rings, in particular the ring of integers modulo $m$, have been studied since the 1970s \cite{blake1972codes,blake1975codes,spiegel1978codes}.
They have, for instance, be used to unify the description of good non-linear binary codes in the Hamming metric, using a connection via  the Gray mapping from linear codes over $\ZZ_4$ with high minimum Lee distance \cite{hammons1994z}.
This Gray mapping was generalized to arbitrary moduli $m$ of $\ZZ_m$ in \cite{constantinescu1997metric}.
Recently, there has been an increased interest in rank-metric codes over finite rings due to the following applications.
\emph{Network coding} over certain finite rings was intensively studied in \cite{feng2014communication,gorla2017algebraic}, motivated
by works on nested-lattice-based network
coding \cite{wilson2010joint,nazer2011compute,feng2013algebraic,tunali2015lattices}
which show that network coding over finite rings may result in more efficient physical-layer network coding schemes.
Kamche et al.~\cite{kamche2019rank} showed how lifted rank-metric codes over finite rings can be used for error correction in network coding. %
The result uses a similar approach as \cite{silva2008rank} to transformation the channel output into a rank-metric error-erasure decoding problem.
Another application of rank-metric codes over finite rings are \emph{space-time codes}.
It was first shown in \cite{kiran2005optimal} how to construct space-time codes with optimal rate-diversity tradeoff via a rank-preserving mapping from rank-metric codes over Galois rings.
This result was generalized to arbitrary finite principal ideal rings in \cite{kamche2019rank}.
The use of finite rings instead of finite fields has advantages since the rank-preserving mapping can be chosen more flexibly.
Kamche et al.~also defined and extensively studied Gabidulin codes over finite principal ideal rings.
In particular, they proposed a Welch--Berlekamp-like decoder for Gabidulin codes and a Gr\"obner-basis-based decoder for interleaved Gabidulin codes \cite{kamche2019rank}. 

Motivated by these recent developments on rank-metric codes over rings, in this paper we define and analyze LRPC codes over Galois rings.
Essentially, we show that Gaborit et al.'s construction and decoder work as well over these rings, with only a few minor technical modifications.
The core difficulty of proving this result is the significantly more involved failure probability analysis, which stems from the weaker algebraic structure of rings compared to fields:
the algorithm and proof are based on dealing with modules over Galois rings instead of vector spaces over finite fields, which behave fundamentally different since Galois rings are usually not integral domains.
We also provide a thorough complexity analysis.
The results can be summarized as follows.

\subsection*{Main Results}
Let $p$ be a prime and $r,s$ be positive integers.
A \emph{Galois ring} $\Rq$ of cardinality $p^{rs}$ is a finite Galois extension of degree $s$ of the ring $\ZZ_{p^r}$ of integers modulo the prime power $p^r$.
As modules over $\Rq$ are not always free (i.e., have a basis), matrices over $\Rq$ have a \emph{rank} and a \emph{free rank}, which is always smaller or equal to the rank.
We will introduce these and other notions formally in Section~\ref{sec:preliminaries}.

In Section~\ref{sec:LRPCcodes}, we construct a family of rank-metric codes and a corresponding family of decoders with the following properties:
Let $m,n,k,\lambda$ be positive integers such that $\lambda$ is greater than the smallest divisor of $m$ and $k$ fulfills $k \leq \tfrac{\lambda-1}{\lambda} n$.
The constructed codes are subsets $\Code \subseteq \Rq^{m \times n}$ of cardinality $|\Code| = |\Rq|^{mk}$. Seen as a set of vectors over an extension ring of $\Rq$, the code is linear w.r.t. this extension ring. We exploit this linearity in the decoding algorithm.

Furthermore, let $t$ be a positive integer with $t < \min\!\left\{\tfrac{m}{\lambda(\lambda+1)/2}, \tfrac{n-k+1}{\lambda}\right\}$. Let $\C \in \Code$ be a (fixed) codeword and let $\E \in \Rq^{m \times n}$ be chosen uniformly at random from all matrices of rank $t$ (and arbitrary free rank). Then, we show in Section~\ref{sec:failure} that the proposed decoder in Section~\ref{sec:decoding} recovers the codeword $\C$ with probability at least
\begin{align*}
1-4 p^{s[\lambda t-(n-k+1)]} - 4 t p^{s\left(t \frac{\lambda(\lambda+1)}{2}-m\right)}.
\end{align*}
Hence, depending on the relation of $p^s$ and $t$, the success probability is positive for
\begin{align*}
t \lessapprox t_\mathrm{max} := \left\lceil\min\!\left\{\tfrac{m}{\lambda(\lambda+1)/2}, \tfrac{n-k+1}{\lambda}\right\}\right\rceil-1.
\end{align*}
and converges exponentially fast to $1$ in the difference $t_\mathrm{max}-t$. Note that for $\lambda=2$ and $m>\tfrac{3}{2}(n-k+1)$, we have $t_\mathrm{max} = \lfloor\tfrac{n-k}{2}\rfloor$.

The decoder has complexity $\softO(\lambda^2 n^2 m)$ operations in $\Rq$ (see Section~\ref{sec:complexity}).
In Section~\ref{sec:simulations}, we present simulation results.

\begin{example}
Consider the case $p=2$, $s=4$, $r=2$, $m=n=101$, $k=40$, and $\lambda=2$.
Then, the decoder in Section~\ref{sec:decoding} can correct up to $t_\mathrm{max} = \lfloor\tfrac{n-k}{2}\rfloor = 30$ errors with success probability at least $1-2^{-6}$. For $t=24$ errors, the success probability is already $\approx 1-2^{-46}$ and for $t=18$, it is $\approx 1-2^{-102}$.
A Gabidulin code as in \cite{kamche2019rank}, over the same ring and the same parameters, can correct any error of rank up to $30$ (i.e., the same maximal radius). However, the currently fastest decoder for Gabidulin codes over rings \cite{kamche2019rank} has a larger complexity than the LRPC decoder in Section~\ref{sec:decoding}.
\end{example}

The results of this paper were partly presented at the IEEE International Symposium on Information Theory 2020 \cite{renner2020lrpc}.
Compared to this conference version, we generalize the results in two ways:
first, we consider LRPC codes over the more general class of Galois rings instead of the integers modulo a prime power.
This is a natural generalization since Galois rings share with finite fields many of the properties needed for dealing with the rank metric. Indeed, they constitute the common point of view  between finite fields and rings of integers modulo a prime power. 
Second, the conference version only derives a bound on the failure probability for errors whose free rank equals their rank.
For some applications, this is no restriction since the error can be designed, but for most communications channels, we cannot influence the error and need to correct also errors of arbitrary rank profile.
Hence, we provide a complete analysis of the failure probability for all types of errors.

\section{Preliminaries}\label{sec:preliminaries}

\subsection{Notation}

Let $A$ be any commutative ring.
We denote modules over $A$ by calligraphic letters, vectors as bold small letters, and matrices as bold capital letters.
We denote the set of $m\times n$ matrices over the ring $A$ by $A^{m\times n}$ and the set of row vectors of length $n$ over $A$ by $A^{n} = A^{1\times n}$. Rows and columns of $m\times n$ matrices are indexed by $1,\hdots,m$ and $1,\hdots,n$, where $X_{i,j}$ denotes the entry in the $i$-th row and $j$-th column of the matrix $\XX$. Moreover, for an element $a$ in a ring $A$, we denote by $\Ann(a)$ the ideal $\Ann(a)=\{b \in A \mid ab=0\}$.

\subsection{Galois Rings}

A Galois ring $\Rq:=\GR(p^r,s)$ is a finite local commutative ring of characteristic $p^r$ and cardinality $p^{rs}$, which is isomorphic to $\ZZ[z]/(p^r,f(z))$, where $f(z)$ is a polynomial of degree $s$ that is irreducible modulo $p$. Let $\maxIdeal$ be the unique maximal ideal of $\Rq$.  It is also well-known that $\Rq$ is a finite chain ring and all its ideals  are powers of $\maxIdeal$ such that $r$ is smallest positive integer $r$ for which $\maxIdeal^r = \{0\}$.
Since Galois rings are principal ideal rings, $\maxIdeal$ is generated by one ring element.
We will call such a generator $g_\maxIdeal$ (which is unique up to invertible multiples).
Note that in a Galois ring  this element can always be chosen to be $p$. 
Moreover, $\Rq/\maxIdeal$ is isomorphic to the finite field 
$\F_{p^s}$.

In this setting, it is well-known that there exists a unique cyclic subgroup of $\Rq^*$ of order $p^s-1$, which is generated by an element $\eta$. The set $T_s := \{0\}\cup \langle \eta\rangle$ is known as \emph{Teichm\"uller set} of $\Rq$.
Every element $a\in \Rq$ has hence a unique representation as
\begin{equation*}
a=\sum_{i=0}^{r-1} g_\maxIdeal^ia_i, \quad a_i\in T_s.
\end{equation*}
We will refer to this as the \emph{Teichm\"uller representation of $a$}. For Galois rings, this representation coincides with the $p$-adic expansion. If, in addition, one chooses the polynomial $h(z)$ to be a \emph{Hensel lift} of a primitive polynomial in $\F_p[x]$ of degree $s$, then the element $\eta$ can be taken to be one of the roots of $h(z)$. Here, for Hensel lift of a primitive polynomial $\bar{h}(z)\in\F_p[z]$, we mean that $h(x)\in\ZZ_{p^r}[z]$ is such that the canonical projection of $h(z)$ over $\F_p[z]$ is $\bar{h}(z)$ and $h(z)$ divides $z^{p^s-1}-1$ in $\ZZ_{p^r}[z]$. The interested reader is referred to \cite{mcdonald1974finite, bini2012finite} for a deeper understanding  on Galois rings.

It is easy to see that the number of units in $\Rq$ is given by
\begin{align}
|\Rq^*| &= |\Rq \setminus \maxIdeal|
= |\Rq| - |\maxIdeal|
= p^{sr} -p^{s(r-1)}
= |\Rq|\big(1-p^{-s}\big). \label{eq:Rq_number_units}
\end{align}

\begin{example}\label{exp:Rq}
  Let $p=2$, $s=1$, $r=3$, and $\Rq = \{0,1,\hdots,7\}$. We have that $\maxIdeal = \{0,2,4,6\}$ and $\Rq/\maxIdeal = \{0,1\} = \F_2$.
  Thus, $g_\maxIdeal = 2$.
  The set $\{1\}$ is the unique cyclic subgroup of $\Rq^*=\{1,3,5,7\}$ of order $p^s-1 = 1$ which is generated by $\eta=1$ and $T_s = \{0,1\}$. Then, the Teichm\"uller representation of $a=5$ is given by $ a = 1\cdot g_\maxIdeal^0 + 0 \cdot g_\maxIdeal^1 + 1 \cdot g_\maxIdeal^2$.
\end{example}

\begin{example}\label{exp:RqGalois}
  Let $p=2$, $s=3$, $r=3$, and let us construct $\Rq = \GR(8,3)$. Consider the ring $\ZZ_{8}$, and $h(z):=z^3+6z^2+5z+7\in\ZZ_{8}[z]$. The canonical projection of the polynomial $h(z)$ over $\F_2[z]$ is $z^3+z+1$ which is primitive, and hence irreducible, in $\F_2[z]$. Thus, we have
  $$\Rq\cong \ZZ_8[z]/(h(z)).$$
  Clearly, $\maxIdeal=(2)\Rq$ and we can choose $g_\maxIdeal=2$. Moreover, if $\eta$ is a root of $h(z)$, then we also have $\Rq\cong \ZZ_8[\eta]$, and every element can be represented as $a_0+a_1\eta+a_2\eta^2$, for $a_0,a_1,a_2\in\ZZ_{8}$.
  On the other hand, the polynomial $h(z)$ divides $x^7-1$ in $\ZZ_8[z]$ and therefore it is a Hensel lift of $z^3+z+1$. This implies that $\eta$ has order $7$, and the Teichm\"uller set is $T_3=\{0,\eta, \eta^2,\ldots,\eta^7=1\}$. If we take the element $a=5+3\eta^2$, then, it can be verified that its Teichm\"uller represntation  is $a=\eta^6+\eta^4g_\maxIdeal+\eta^5g_\maxIdeal^2=\eta^6+\eta^4\cdot 2+\eta^5\cdot 4$.
 \end{example}

\subsection{Extensions of Galois rings}

Let  $h(z) \in \Rq[z]$  be a polynomial of degree $m$ such that the leading coefficient of $h(z)$ is a unit and $h(z)$ is irreducible over the finite field $\Rq/\maxIdeal$. Then, the Galois ring $\Rq[z]/(h(z))$ is denoted by $\Rqm$. We have that $\Rqm$ is the Galois ring $\GR(p^r,sm)$, with maximal ideal $\MaxIdeal = \maxIdeal \Rqm$. Moreover, it is known that subrings of Galois rings are Galois rings and that for every $\ell$ dividing $m$ there exists a unique subring of $\Rqm$ which is a Galois extension of degree $\ell$ of $\Rq$.
These are all subrings of $\Rqm$ that contain $\Rq$.
In particular there exists a unique copy of $\Rq$ in $\Rqm$, and we can therefore consider (with a very small abuse of notation) $\Rq \subseteq \Rqm$. In particular, we have that $g_\maxIdeal$ is also the generator of $\MaxIdeal$ in $\Rqm$. 

As for $\Rq$, also $\Rqm$ contains a unique cyclic subgroup of order $p^{sm}-1$, and we can consider the Teichm\"uller set $T_{sm}$ as the union of such a subgroup together with the $0$ element. Hence, every $a\in\Rqm$ has a unique representation as
$$ a=\sum_{i=0}^{r-1} g_\maxIdeal^ia_i, \quad a_i\in T_{sm}.$$
The number of units in $\Rqm$ is given by
\begin{align*}
|\Rqm^*| &= |\Rqm \setminus \MaxIdeal| 
= |\Rqm| - |\MaxIdeal| 
= p^{srm} - |\maxIdeal|^m 
= p^{srm} - \big(p^{s(r-1)}\big)^m \\
&= p^{srm}\big(1-p^{-sm}\big) 
= |\Rqm|\big(1-p^{-sm}\big).
\end{align*}

From now on and for the rest of the paper, we will always denote by $\Rq$ the Galois ring $\GR(p^r,s)$, and by $\Rqm$ the Galois ring $\GR(p^r,sm)$.

\subsection{Smith Normal Form}

The Smith normal form is well-defined for both $\Rq$ and $\Rqm$, i.e., for $\A \in \Rq^{m \times n}$, there are invertible matrices $\S \in \Rq^{m \times m}$ and $\T \in \Rq^{n \times n}$ such that
\begin{align*}
\D = \S \A \T \in \Rq^{m \times n}
\end{align*}
is a diagonal matrix with diagonal entries $d_1,\dots,d_{\min\{n,m\}}$ with
\begin{align*}
d_j \in \maxIdeal^{i_j} \setminus \maxIdeal^{i_j+1},
\end{align*}
where the $0 \leq i_1 \leq i_2 \leq \dots \leq i_{\min\{n,m\}} \leq r$.
The same holds for matrices over $\Rqm$, where we replace $\maxIdeal$ by $\MaxIdeal$ (note that $\MaxIdeal^r=\{0\}$ and $\MaxIdeal^{r-1}\neq \{0\}$ for the \emph{same} $r$). The rank and the free rank of $\A$ (w.r.t.\ a ring $A \in \{\Rqm,\Rq\}$) is defined by $\rk (\A) := |\{ i\in\{1,\hdots,\min\{m,n\}\}: \D_{i,i} \not = 0 \}|$ and $\frk (\A) := |\{ i \in \{1,\hdots,\min\{m,n\}\} :\D_{i,i} \text{ is a unit} \}|$, respectively, where $\D$ is the diagonal matrix of the Smith normal form w.r.t.\ the ring $R$.

\subsection{Modules over Finite Chain Rings}

The ring $\Rqm$ is a free module over $\Rq$ of rank $m$. Hence, elements of $\Rqm$ can be treated as vectors in $\Rq^m$ and linear independence, $\Rq$-subspaces of $\Rqm$ and the $\Rq$-linear span of elements are well-defined. Let $\ve{\gamma}=[\gamma_1,\hdots,\gamma_m]$ be an ordered basis of $\Rqm$ over $\Rq$. By utilizing the module space isomorphism $\Rqm \cong \Rq^m$, we can relate each vector $\a \in \Rqm^{n}$ to a matrix $\A \in \Rq^{m\times n}$ according to $\extsmallfield_{\gamma} : \Rqm^{n} \rightarrow \Rq^{m\times n}, \a \mapsto \A$, where $a_j = \sum_{i=1}^{m} A_{i,j} \gamma_{i}$,\ $j \in \{1,\hdots,n\}$.
The \emph{ (free) rank norm} $(\ftemp)\rk_{\Rq}(\a)$ is the (free) rank of the matrix representation $\A$, \emph{i.e.}, $\rk_{\Rq}(\a) := \rk(\A)$ and $\frk_{\Rq}(\a) := \frk(\A)$, respectively.
\begin{example}
 Let $p=2$, $s=1$, $r=3$ as in Example~\ref{exp:Rq}, $h(z) = z^3+z+1$  and
  \begin{equation*}
    \a =
    \begin{bmatrix}
 2z^2 + 2z + 5, &  4z^2 + z + 6,  &     2z^2 + z
    \end{bmatrix}.
\end{equation*}
Using a polyomial basis $\ve{\gamma}=[1,z,z^2]$, the matrix representation of $\a$ is
\begin{equation*}
  \A =
  \begin{bmatrix}
    5 & 6 & 0\\
    2 & 1 & 1\\
    2 & 4 & 2
  \end{bmatrix}
\end{equation*}
and the Smith normal form of $\A$ is given by
\begin{equation*}
  \D = \begin{bmatrix}
    1 & 0 & 0 \\
    0 & 1 & 0 \\
    0 & 0 & 2 \\
 \end{bmatrix}.
\end{equation*}
It can be observed that $d_1, d_2 \in \maxIdeal^0 \setminus \maxIdeal^1 = \{1,3,5,7\}$ and $d_3 \in \maxIdeal^1 \setminus \maxIdeal^2 = \{2,6\}$ and thus $\rk(\A) = \rk(\D)  = 3$ and $\frk(\A)= \frk(\D)= 2$. It follows that $\rk_{\Rq}(\a) = 3$ and $\frk_{\Rq}(\a) = 2$.
\end{example}
Let $a = \sum_{i=1}^{m} a_i \gamma_i \in \Rqm$, where $a_i \in \Rq$. The following statements are equivalent (cf.~\cite[Lemma~2.4]{kamche2019rank}):
\begin{itemize}
\item $a$ is a unit in $\Rqm$.
\item At least one $a_i$ is a unit in $\Rq$.
\item $\{a\}$ is linearly independent over $\Rq$.
\end{itemize}

The $\Rq$-linear module that is spanned by $v_1,\hdots,v_{\ell} \in \Rqm$ is denoted by $\langle v_1,\dots,v_\ell \rangle_{\Rq} := \big\{\sum_{i=1}^{\ell} a_i v_i : a_i \in \Rq \big\}$. The $\Rq$-linear module that is spanned by the entries of a vector $\a \in \Rqm^{n}$ is called the support of $\a$, \emph{i.e.}, $\supp(\a) := \langle a_1,\dots,a_n \rangle_{\Rq}$. Further, $\Aspace \cdot \Bspace$ denotes the product module of two submodules $\Aspace$ and $\Bspace$ of $\Rqm$, i.e., $\Aspace \cdot \Bspace := \langle a \cdot b \, : \, a \in \Aspace, \, b \in \Bspace \rangle$.

\subsection{Valuation in Galois Rings}

We define the \emph{valuation} of $a \in \Rq \setminus \{0\}$ as the unique integer $v(a) \in \{0,\dots,r-1\}$ such that
\begin{align*}
a \in \maxIdeal^{v(a)} \setminus \maxIdeal^{v(a)+1},
\end{align*}
and set $v(0) := r$.
In the same way, the \emph{valuation} of $b \in \Rqm \setminus \{0\}$ as the unique integer $v(b) \in \{0,\dots,r-1\}$ such that
\begin{align*}
b \in \MaxIdeal^{v(b)} \setminus \MaxIdeal^{v(b)+1},
\end{align*}
and $v(0) = r$.

 Let $\{\gamma_1,\ldots, \gamma_m\}$ be a basis of $\Rqm$ as $\Rq$-module. It is easy to see that for $a = \sum_{i=1}^{m} a_i \gamma_i \in \Rqm \setminus \{0\}$, where $a_i \in \Rq$ (not all $0$), we have
\begin{equation}\label{eq:valuation}
v(a) = \min_{i=1,\dots,m}\{v(a_i)\}.
\end{equation}

\begin{example}
 Let $p=2$, $s=1$, $r=3$ as in Example~\ref{exp:Rq}, $h(z) = z^3+z+1$  and let $a=1$, $b=2$, $c=4 \in \Rq$. Since $a \in \maxIdeal^0 \setminus \maxIdeal^1=\{1,3,5,7\}$, $b \in \maxIdeal^1 \setminus \maxIdeal^2 = \{2,6\}$, and $c \in \maxIdeal^2 \setminus \maxIdeal^3=\{4\}$, one obtains $v(a) =0$, $v(b) = 1$ and $v(c)=2$.

  Furthermore, let $d=2z^2+1$, $e=4z^2+2z+2$, $f=4z^2+4$, where $d\in\MaxIdeal^0\setminus\MaxIdeal^1$, $e\in \MaxIdeal^1\setminus\MaxIdeal^2$ and $f\in\MaxIdeal^2\setminus\MaxIdeal^3$. It follows that $v(d)=0$, $v(e)=1$ and $v(f)=2$. Since an element is a unit if and only if its valuation is equal to $0$, only the elements $a$ and $d$ are units.
\end{example}

\subsection{Rank Profile of a Module and \Mingensets}

Let $\Mspace$ be an $\Rq$-submodule of $\Rqm$ and $d_1,\dots,d_n$ be diagonal entries of a Smith normal form of a matrix whose row space is $\Mspace$.
Define the \emph{rank profile of $\Mspace$} to be the polynomial
\begin{align*}
\phi^{\Mspace}(x) := \sum_{i=0}^{r-1} \phi_i^{\Mspace} x^i \in \ZZ[x]/(x^r),
\end{align*}
where
\begin{equation*}
\phi^{\Mspace}_i := \left|\left\{j : v(d_j)=i\right\}\right|.
\end{equation*}
Note that $\phi^{\Mspace}(x)$ is independent of the chosen matrix and Smith normal form since the diagonal entries $d_i$ are unique up to multiplication by a unit.
We can easily read the free rank and rank from the rank profile
\begin{align*}
\frk_{\Rq} \Mspace &= \phi^{\Mspace}_0 = \phi^{\Mspace}(0), \\
\rk_{\Rq} \Mspace &= \sum_{i=0}^{r-1} \phi^{\Mspace}_i = \phi^{\Mspace}(1).
\end{align*}

\begin{example}
Consider the ring $\Rq=\GR(8,3)$ as defined in Example \ref{exp:RqGalois}, where as generator of $\maxIdeal$ we take $g_\maxIdeal=2$.  Take a module $\Mspace$ whose diagonal matrix in the Smith normal form is
\begin{align*}
\begin{bmatrix}
1 &   & & &\\
  & 1 & & & \\
  &   & 2 & & \\
  &   &   & 4 &\\
  &   &   &  & 0
\end{bmatrix}. 
\end{align*}
We have
\begin{equation*}
\phi^{\Mspace}(x) = 2+x+x^2.
\end{equation*}
\end{example}

On $\ZZ[x]/(x^r)$, we define the following partial order $\preceq$.
\begin{definition}\label{def:ordering_rankprofiles}
Let $a(x),b(x) \in \ZZ[x]/(x^r)$.
We say that $a(x) \preceq b(x)$ if for every $i\in \{0,\ldots, r-1\}$ we have
\begin{equation*}
\sum_{j=0}^i a_j \leq \sum_{j=0}^i b_j.
\end{equation*}
\end{definition}

\begin{remark}
The partial order $\preceq$ on rank profiles is compatible with the containment of submodules. That is, if $M_1\subseteq M_2$ then $\phi^{\Mspace_1} \preceq \phi^{\Mspace_2}$. Clearly the opposite implication is not true in general.
\end{remark}

For $\D$ and $\T$ as in the Smith normal form of a matrix over $\Rq$, observe that the nonzero rows of the matrix $\D\T^{-1}$ produce a set of generators for the $\Rq$-module generated by the rows of $\A$, which is minimal and of the form
$$\Gamma=\{g_\maxIdeal^ia_{i,\ell_i} \mid 0\leq i \leq r-1, 1 \leq \ell_i \leq \phi^{\Mspace}_i \}.$$ 
A generating set coming from the Smith Normal Form as described above will be called \emph{\mingenset}. Alternatively, a \mingenset for a $\Rq$-module $\Mspace$ is a generating set $\{b_{i,\ell_i} \mid 0\leq i \leq r-1, 1 \leq \ell_i \leq \phi^{\Mspace}_i \}$ such that $v(b_{i,\ell_i})=i$. 
Moreover,  every $\Rq$-submodule of $\Rq^n$ can be seen as the rowspace of a matrix, and hence it decomposes as 
$$ \Mspace=\langle\Gamma^{(0)}\rangle_{\Rq}+\maxIdeal \langle\Gamma^{(1)}\rangle_{\Rq} + \ldots + \maxIdeal^{r-1}\langle\Gamma^{(r-1)}\rangle_{\Rq}, $$
where  $\Gamma^{(i)}:=\{a_{i,\ell_i} \mid 0\leq i \leq r-1, 1 \leq \ell_i \leq \phi^{\Mspace}_i \}$. It is easy to see that $\langle\Gamma^{(i)}\rangle_{\Rq}$ is a free module. However, this decomposition depends on the chosen \mingenset $\Gamma$.

For a module $\Mmodule$ with \mingenset $\Gamma = \{g_\maxIdeal^ia_{i,\ell_i} \mid 0\leq i \leq r-1, 1 \leq \ell_i \leq \phi^{\Mspace}_i \}$, we have the following: Let $e \in \Mmodule$ and
\begin{align*}
e = \sum_{i=0}^{r-1} \sum_{\ell_i=1}^{\phi_i^\Mmodule} e_{i,\ell_i} g_\maxIdeal^ia_{i,\ell_i} = \sum_{i=0}^{r-1} \sum_{\ell_i=1}^{\phi_i^\Mmodule} e'_{i,\ell_i} g_\maxIdeal^i a_{i,\ell_i}
\end{align*}
be two different representations of $e$ in the \mingenset with coefficients $e_{i,\ell_i},e'_{i,\ell_i} \in \Rq$, respectively. Then, we have
\begin{align*}
e_{i,\ell_i} \equiv e'_{i,\ell_i} \mod g_\maxIdeal^{r-i}
\end{align*}
for all $0\leq i \leq r-1$ and $1 \leq \ell_i \leq \phi^{\Mspace}_i$. This is due to the fact that by definition of \mingenset, the set $\{a_{i,\ell_i} \mid 0\leq i \leq r-1, 1 \leq \ell_i \leq \phi^{\Mspace}_i\}$ is linear independent over $\Rq$, and hence  $(e_{i,\ell_i} - e'_{i,\ell_i}) g_\maxIdeal^{i}=0$ for every $i,\ell_i$. Therefore, the representation of an element in $\Mspace$ with respect to a \mingenset have uniquely determined coefficients $e_{i,\ell_i}$ modulo $\Ann(g_\maxIdeal^i)=\maxIdeal^{r-i}$.

\begin{lemma}\label{lem:shiftedrank}
Let $\Mspace$ be an $\Rq$-submodule of $\Rqm$ with rank-profile $\phi^{\Mspace}$ and let $j \in \{1,\ldots,r-1\}$. Then, the rank-profile of $\maxIdeal^j\Mspace$ is given by
\begin{equation*}
\phi^{\maxIdeal^j\Mspace}(x) = x^j \phi^{\Mspace}(x).
\end{equation*}
In particular, the rank of $\maxIdeal^j\Mspace$ is equal to $\phi^{\maxIdeal^j\Mspace}(1)=\sum\limits_{i=0}^{r-1-j}\phi^{\Mspace}_i$.
\end{lemma}

\begin{proof}
Let $g_\maxIdeal$ be a generator of $\maxIdeal$.
If $\Gamma=\{g_\maxIdeal^ia_{i,\ell_i} \mid 0\leq i \leq r-1, 1 \leq \ell_i \leq \phi^{\Mspace}_i \}$ is a \mingenset for $M$, then it is easy to see that 
$$\left\{  g_\maxIdeal^{i+j}a_{i,\ell_i} \mid 0\leq i \leq r-j-1, 1 \leq \ell_i \leq \phi^{\Mspace}_i \right\}$$ is a \mingenset for $\maxIdeal^j\Mspace$. Hence, the first $j$ coefficients of $\phi^{\maxIdeal^j\Mspace}(x)$ are equal to zero, while the remaining ones are the $j$-th shift of the first $r-j$ coefficients of $\phi^{\Mspace}(x)$. \qed
\end{proof}

\begin{proposition}\label{prop:productpolynomials}
For any pair  of $\Rq$-submodules $\Mspace_1, \Mspace_2$ of $\Rqm$, we have 
\begin{equation*}
\phi^{\Mspace_1 \cdot \Mspace_2}(x) \preceq \phi^{\Mspace_1}(x) \phi^{\Mspace_2}(x).
\end{equation*}
\end{proposition}

\begin{proof}
Let $g_\maxIdeal$ be a generator of $\maxIdeal$.
Let $\Mspace_1, \Mspace_2$ be two $\Rq$-submodules  with rank-profile $\phi^{\Mspace_1}$ and $\phi^{\Mspace_2}$ respectively. Then, there exist a minimal generating set of $\Mspace_1$ given by
$$\Gamma_1:=\{g_\maxIdeal^ia_{i,j_i} \mid 0\leq i \leq r-1, 1 \leq j_i \leq \phi^{\Mspace_1}_i \},$$
and a minimal generating set of $M_2$ given by
$$\Gamma_2:=\{g_\maxIdeal^ib_{i,j_i} \mid 0\leq i \leq r-1, 1 \leq j_i \leq \phi^{\Mspace_2}_i \}.$$
In particular, the product set $\Gamma_1\cdot \Gamma_2$ is a generating set of $\Mspace_1 \cdot \Mspace_2$. Hence
\begin{align*}
\sum_{i=0}^{r-1} \phi^{\Mspace_1\cdot \Mspace_2}_i &=\rk_{\Rq}(\Mspace_1\cdot \Mspace_2) \\
& \leq |\Gamma_1 \cdot \Gamma_2\setminus\{0\}| \\
&= \sum_{i=0}^{r-1}\sum_{j=0}^i \phi^{\Mspace_1}_j\phi^{\Mspace_2}_{i-j}\\
&=\sum_{i=0}^{r-1}(\phi^{\Mspace_1} \phi^{\Mspace_2})_i.
\end{align*}
The general inequality for the truncated sums then follows by considering the rank of the submodule $\maxIdeal^j(\Mspace_1 \cdot \Mspace_2)$ and Lemma \ref{lem:shiftedrank}.
\qed
\end{proof}

\section{LRPC Codes Over Galois Rings}\label{sec:LRPCcodes}

\begin{definition}\label{def:LRPCcodes}
Let $k,n,\lambda$ be positive integers with $0<k<n$. %
Furthermore, let $\Fspace \subseteq \Rqm$ be a free $\Rq$-submodule of $\Rqm$ of rank $\lambda$.
A low-rank parity-check (LRPC) code with parameters $\lambda,n,k$ is a code with a parity-check matrix
$\H \in \Rqm^{(n-k) \times n}$
such that $\rank_{\Rqm} \H = \frk_{\Rqm} \H = n-k$ and $\Fspace = \langle H_{1,1},\dots,H_{(n-k),n} \rangle_{\Rq}$.
\end{definition}

Note that an LRPC code is a free submodule of $\Rqm^n$ of rank $k$.
This means that the cardinality of the code is $|\Rqm|^k = |\Rq|^{mk} = p^{r s m k}$.
We define the following three additional properties of the parity-check matrix that we will use throughout the paper to prove the correctness of our decoder and to derive failure probabilities.
As for rank-metric codes over finite fields, we can interpret vectors over $\Rqm$ as matrices over $\Rq$ by the $\Rq$-module isomorphism $\Rqm \simeq \Rq^m$.
In particular, an LRPC code can be seen as a subset of $\Rq^{m \times n}$.

\begin{definition}\label{def:H_properties}
Let $\lambda$, $\Fspace$, and $\H$ be defined as in Definition~\ref{def:LRPCcodes}.
Let $f_1,\dots,f_\lambda \in \Rqm$ be a free basis of $\Fspace$.
For $i=1,\dots,n-k$, $j=1,\dots,n$, and $\ell=1,\dots,\lambda$, let $h_{i,j,\ell} \in \Rq$ be the unique elements such that $H_{i,j} = \sum_{\ell = 1}^{\lambda} h_{i,j,\ell} f_{\ell}$.
Define
\begin{equation}
\H_{\mathrm{ext}} :=
\begin{bmatrix}
h_{1,1,1} & h_{1,2,1} & \hdots & h_{1,n,1} \\
h_{1,1,2} & h_{1,2,2} & \hdots & h_{1,n,2} \\
\vdots & \vdots & \ddots & \vdots \\
h_{2,1,1} & h_{2,2,1} & \hdots & h_{2,n,1} \\
h_{2,1,2} & h_{2,2,2} & \hdots & h_{2,n,2} \\
\vdots & \vdots & \ddots & \vdots \\
\end{bmatrix}
\in \Rq^{(n-k)\lambda \times n}.
\label{eq:H_ext}
\end{equation}
Then, $\H$ has the
\begin{enumerate}
\item \textbf{unique-decoding property} if 
$\lambda \geq \tfrac{n}{n-k}$ and $\frk \left( \H_{\mathrm{ext}} \right) = \rk \left(  \H_{\mathrm{ext}} \right) = n$,

\item \textbf{maximal-row-span property} if every row of the parity-check matrix $\H$ spans the entire space $\Fspace$,
\item \textbf{unity property} if every entry $H_{i,j}$ of $\H$ is chosen from the set
$H_{i,j} \in \tilde{\Fspace} := \left \{ \textstyle\sum_{i=1}^{\lambda} \alpha_i f_i \, : \, \alpha_i \in \Rq^* \cup \{0\} \right\}  \subseteq \Fspace$.
\end{enumerate}
Furthermore, we say that $\Fspace$ has the \textbf{\Fcontainsoneproperty} if $1 \in \Fspace$.
\end{definition}

In the original papers about LRPC codes over finite fields, \cite{gaborit2013low,aragon2019low}, some of the properties of Definition~\ref{def:H_properties} are used without explicitly stating them.

We will see in Section~\ref{ssec:erasure_decoding} that the {unique-decoding property} together with a property of the error guarantees that erasure decoding always works (i.e., that the full error vector can be recovered from knowing the support and syndrome of an error).
This property is also implicitly used in \cite{gaborit2013low}.
It is, however, not very restrictive: if the parity-check matrix entries $H_{i,j}$ are chosen uniformly at random from $\Fspace$, this property is fulfilled with the probability that a random $\lambda (n-k) \times n$ matrix has full (free) rank $n$.
This probability is arbitrarily close to $1$ for increasing difference of $\lambda(n-k)$ and $n$ (cf.~\cite{renner2019efficient} for the field and Lemma~\ref{lem:syndrome_number_of_full-rank_matrices} in Section~\ref{ssec:failure_syndrome} for the ring case).

We will use the {maximal-row-span property} to prove a bound on the failure probability of the decoder in Section~\ref{sec:failure}. It is a sufficient condition that our bound (in particular Theorem~\ref{thm:syndrome_condition_main_statement} in Section~\ref{sec:failure}) holds. Although not explicitly stated, \cite[Proposition~4.3]{aragon2019low} must also assume a similar or slightly weaker condition in order to hold. It does not hold for arbitrary parity-check matrices as in \cite[Definition~4.1]{aragon2019low} (see the counterexample in Remark~\ref{rem:necessity_of_maximal_row_span_condition_or_similar} in Section~\ref{sec:failure}). This is again not a big limitation in general for two reasons: first, the ideal codes in \cite[Definition~4.2]{aragon2019low} appear to automatically have this property, and second, a random parity-check matrix has this property with high probability.

In the case of finite fields, the {unity property} is no restriction at all since the units of a finite field are all non-zero elements. That is, we have $\tilde{\Fspace} = \Fspace$. Over rings, we need this additional property as a sufficient condition for one of our failure probability bounds (Theorem~\ref{thm:syndrome_condition_main_statement} in Section~\ref{sec:failure}). It is not a severe restriction in general, since 
\begin{align*}
\frac{|\tilde{\Fspace}|}{|\Fspace|} = \frac{(|\Rq^*|+1)^\lambda}{|\Rq|^\lambda} = \big(1-p^{-s}+p^{-sr}\big)^\lambda,
\end{align*}
which is relatively close to $1$ for large $p^s$ and comparably small $\lambda$.

Finally, Gaborit et al.\ \cite{gaborit2013low} also used the \Fcontainsoneproperty of $\Fspace$.
In contrast to the other three properties in Definition~\ref{def:H_properties}, this property only depends on $\Fspace$ and not on $\H$.
We will also assume this property to derive a bound on the probability of one possible cause of a decoding failure event in Section~\ref{ssec:failure_intersection}.

\section{Decoding}\label{sec:decoding}

\subsection{The Main Decoder}

Fix $\lambda$ and $\Fspace$ as in Definition~\ref{def:LRPCcodes}.
Let $f_1,\dots,f_\lambda \in \Rqm$ be a free basis of $\Fspace$.
Note that since the $f_i$ are linearly independent, the sets $\{f_i\}$ are linearly independent, which by the discussion in Section~\ref{sec:preliminaries} implies that all the $f_i$ are units in $\Rqm$. Hence, $f_i^{-1}$ exists for each $i$.
We will discuss erasure decoding (Line~\ref{line:erasure_decoding}) in Section~\ref{ssec:erasure_decoding}.

\begin{algorithm}
\DontPrintSemicolon
\caption{LRPC Decoder}
\label{alg:decoder}
\KwIn{\begin{itemize}
\item LRPC parity-check matrix $\H$ (as in Definition~\ref{def:LRPCcodes})
\item $\r = \c + \e$, such that
\begin{itemize}
\item $\c$ is in the LRPC code $\mathcal{C}$ given by $\H$ and
\item The support of $\e$ is a module of rank $t$. %
\end{itemize}
\end{itemize}}
\KwOut{Codeword $\c'$ of $\mathcal{C}$ or ``decoding failure''}
$\s = [s_1,\dots,s_{n-k}] \gets \r \H^\top$ \label{line:syndrome_computation} \\
$\Sspace \gets \langle s_1,\dots,s_{n-k} \rangle_{\Rq}$ \label{line:Sspace} \\
\For{$i=1,\dots,\lambda$}{
$\Sspace_i \gets f_i^{-1} \Sspace = \left\{f_i^{-1} a \, : \, a \in \Sspace \right\}$ \label{line:S_i_computation} \\
}
$\Espace' \gets \bigcap_{i=1}^{\lambda} \Sspace_i$ \label{line:S_i_intersection} \\
$\e \gets$ Erasure decoding with support $\Espace'$ w.r.t.~the syndrome $\s$, as described in Lemma~\ref{lem:erasure_decoding} (Section~\ref{ssec:erasure_decoding}) \label{line:erasure_decoding} \\
\If{There is exactly one solution $\e$ of the erasure decoding problem}{
\Return{$\r-\e$}
} \Else {
\Return{``decoding failure''}
}
\end{algorithm}

Algorithm~\ref{alg:decoder} recovers the support $\Espace$ of the error $\e$ if $\Espace' = \Espace$. A necessary (but not sufficient) condition for this to be fulfilled is that we have $\Sspace = \Espace \cdot \Fspace$.
Furthermore, we will see in Section~\ref{ssec:erasure_decoding} that we can uniquely recover the error vector $\e$ from its support $\Espace$ and syndrome $\s$ if the the parity-check matrix fulfills the unique decoding property and we have $\phi^{\Espace \cdot \Fspace} = \phi^{\Espace} \phi^{\Fspace}$.
Hence, decoding works if the following three conditions are fulfilled:
\begin{enumerate}
\item \label{itm:fail_product_small} 	$\phi^{\Espace \cdot \Fspace} = \phi^{\Espace} \phi^{\Fspace}$, \hfill (\textbf{product condition}).
\item \label{itm:fail_Sspace_small} 	$\Sspace = \Espace \cdot \Fspace$, \hfill (\textbf{syndrome condition}) 
\item \label{itm:fail_intersection_big} $\bigcap_{i=1}^{\lambda} \Sspace_i = \Espace$, \hfill (\textbf{intersection condition}), 
\end{enumerate}

We call the case that at least one of the three conditions is not fulfilled a \emph{(decoding) failure}.
We will see in the next section (Section~\ref{sec:failure}) that whether an error results in a failure depends solely on the error support $\Espace$.
Furthermore, given an error support that is drawn uniformly at random from the modules of a given rank profile $\phi$, the failure probability can be upper-bounded by a function that depends only on the rank of the module (i.e., $\phi^\Espace(1)$).

In Section~\ref{sec:complexity}, we will analyze the complexity of Algorithm~\ref{alg:decoder}.
The proofs in that section also indicate how the algorithm can be implemented in practice.

\begin{remark}
Note that the success conditions above imply that for an error of rank $\phi^\Espace(1) = t$, we have $\lambda t \leq m$ (due to the {product condition}) as well as $\lambda \geq \tfrac{n}{n-k}$ (due to the {unique-decoding property}). Combined, we obtain
$t \leq m\tfrac{n-k}{n} = m(1-R)$,
where $R := \tfrac{k}{n}$ is the rate of the LRPC code. %
\end{remark}

\subsection{Erasure Decoding}\label{ssec:erasure_decoding}

As its name suggests, the unique decoding property of the parity-check matrix is related to unique erasure decoding, \emph{i.e.}, the process of obtaining the full error vector $\e$ after having recovered its support. The next lemma establishes this connection.

\begin{lemma}[Unique Erasure Decoding]\label{lem:erasure_decoding}
Given a parity-check matrix $\H$ that fulfills the {unique-decoding property}. Let $\Espace$ be a free support of rank $t \leq \tfrac{m}{\lambda}$. If $\phi^{\Espace \cdot\Fspace} = \phi^\Espace \phi^\Fspace$, then, for any syndrome $\s \in \Rqm^{n-k}$, there is at most one error vector $\e \in \Rqm^n$ with support $\Espace$ that fulfills
$\H \e^\top = \s^\top$.
\end{lemma}

\begin{proof}
Let $f_1,\dots,f_\lambda$ be a basis of the free module $\Fspace$. Furthermore, let $\varepsilon_1,\dots,\varepsilon_t$ be an \mingenset of $\Mmodule$. To avoid too complicated sums in the derivation below, we use a slightly different notation as in the definition of \mingenset and write $\varepsilon_j = g_\maxIdeal^{v(\varepsilon_j)} \varepsilon_j^*$ for all $j=1,\dots,t$, where $\varepsilon^*_j \in \Rqm^*$ are units.

Due to $\phi^{\Espace \cdot \Fspace} = \phi^\Espace \phi^\Fspace$, we have that $f_i \varepsilon_\kappa$ for $i=1,\dots,\lambda$ and $\kappa=1,\dots,t$ is an \mingenset of the product space $\Espace\cdot\Fspace$.
Any entry of the parity-check matrix $\H$ has a unique representation $H_{i,j} = \sum_{\ell = 1}^{\lambda} h_{i,j,\ell} f_{\ell}$ for $h_{i,k,\ell} \in \Rq$.
Furthermore, any entry of error vector $\e = [e_1,\dots,e_n]$ can be represented as $e_j = \sum_{\kappa=1}^{t} e_{j,\kappa} \varepsilon_\kappa$, where the $e_{j,\kappa} \in \Rq$ are unique modulo $\maxIdeal^{r-v(\varepsilon_\kappa)}$.

We want to recover the error vector $\e$ from the syndrome $\s = [s_1,\dots,s_{n-k}]^\top$, which are related by definition as follows:
\begin{align*}
s_{i} &=\sum_{j=1}^{n}H_{i,j}e_{j} \\
&=\sum_{j=1}^{n}\sum_{\ell=1}^{\lambda}h_{i,j,\ell}f_{\ell}\sum_{\kappa=1}^{t}e_{j,\kappa}\varepsilon_{\kappa} \\
&=\sum_{j=1}^{n}\sum_{\ell=1}^{\lambda} \underbrace{\sum_{\kappa=1}^{t}h_{i,j,\ell}e_{j,\kappa}}_{=: \, s_{i,\ell,\kappa}} f_{\ell}\varepsilon_{\kappa} \\
&=\sum_{\ell=1}^{\lambda}\sum_{\kappa=1}^{t}s_{i,\ell,\kappa} f_{\ell}\varepsilon_{\kappa}. 
\end{align*}
Hence, for any representation $e_{j,\kappa}$ of the error $\e$, there is a representation $s_{i,\ell,\kappa}$ of $\s$. If we know the latter representation, it is easy to obtain the corresponding $e_{j,\kappa}$ under the assumed conditions: write
\begin{align*}
s_{i,\ell,\kappa} = \sum_{j=1}^{n}h_{i,j,\ell} e_{j,\kappa},\quad \ell=1,\dots,\lambda, \, \kappa=1,\dots,t, \,  i=1,\dots,n-k.
\end{align*}
We can rewrite this into $t$ independent linear systems of equations of the form
\begin{align}
\underbrace{\begin{bmatrix}
s_{1,1,\kappa} \\
s_{1,2,\kappa} \\
\vdots \\
s_{2,1,\kappa} \\
s_{2,2,\kappa} \\
\vdots
\end{bmatrix}}_{=: \, \s^{(\kappa)}} = \H_{\mathrm{ext}} \cdot \underbrace{\begin{bmatrix}
e_{1,\kappa} \\
e_{2,\kappa} \\
\vdots \\
e_{n,\kappa} \\
\vdots
\end{bmatrix}}_{=: \, \e^{(\kappa)}} \label{eq:erasure_decoding_system}
\end{align}
for each $\kappa=1,\dots,t$, where $\H_{\mathrm{ext}} \in \Rq^{(n-k)\lambda \times n}$ is independent of $\kappa$ and defined as in \eqref{eq:H_ext}.

By the unique decoding property, $\H_{\mathrm{ext}}$ has more rows than columns (i.e, $(n-k)\lambda\geq n$) and full free rank and rank (equal to $n$). Hence, each system in \eqref{eq:erasure_decoding_system} has a unique solution $\e^{(\kappa)}$.

It is left to show that any representation $s_{i,\ell,\kappa}$ of $\s$ in the \mingenset $f_i \varepsilon_\kappa$ of $\Espace \cdot \Fspace$ yields the same error vector $\e$.
Recall that $s_{i,\ell,\kappa}$ is unique modulo $\maxIdeal^{r-v(\varepsilon_i)}$ (note that $v(f_i \varepsilon_\kappa) = v(\varepsilon_\kappa)$).
Assume now that we have a different representation, say
\begin{equation*}
{\s'}^{(\kappa)} = \s^{(\kappa)} + g_{\maxIdeal}^{r-v(\varepsilon_\kappa)} \ve{\chi},
\end{equation*}
where $\ve{\chi} \in \Rq^{(n-k)\lambda}$. Then the unique solution ${\e'}^{(\kappa)}$ of the linear system ${\s'}^{(\kappa)} \H_\mathrm{ext} {\e'}^{(\kappa)}$ is of the form
\begin{align*}
{\e'}^{(\kappa)} = \e^{(\kappa)} + g_{\maxIdeal}^{r-v(\varepsilon_\kappa)} \ve{\mu}
\end{align*}
for some $\ve{\mu'} \in \Rq^{(n-k)\lambda}$.
Hence, ${\e'}^{(\kappa)} \equiv \e^{(\kappa)} \mod \maxIdeal^{r-v(\varepsilon_\kappa)}$, which means that the two representations ${\e'}^{(\kappa)}$ and $\e^{(\kappa)}$ belong to the same error~$\e$.

This shows that we can take \emph{any} representation of the syndrome vector $\s$, solve the system in \eqref{eq:erasure_decoding_system} for $\e^{(\kappa)}$ for $\kappa=1,\dots,t$, and obtain the unique error vector $\e$ corresponding to this syndrome $\s$ and support $\Espace$.
\qed
\end{proof}

\section{Failure Probability}\label{sec:failure}

Consider an error vector $\e$ that is chosen uniformly at random from the set of error vectors whose support is a module of a given rank profile $\phi \in \mathbb Z[x]/(x^r)$  and rank $\phi(1) = t$.
In this section, we derive a bound on the failure probability of the LRPC decoder over Galois rings for this error model.
The resulting bound does not depend on the whole rank profile $\phi$, but only on the rank $t$.

This section is the most technical and involved part of the paper.
Therefore, we derive the bound in three steps, motivated by the discussion on failure conditions in Section~\ref{sec:decoding}:
In Section~\ref{ssec:failure_product}, we derive an upper bound on the failure probability of the product condition.
Section~\ref{ssec:failure_syndrome} presents a bound on the syndrome condition failure probability conditioned on the event that the product condition is fulfilled.
Finally, in Section~\ref{ssec:failure_intersection}, we derive a bound on the intersection failure probability, given that the first conditions are satisfied.

The proof strategy is similar to the analogous derivation for LRPC codes over fields by Gaborit et al.~\cite{gaborit2013low}. However, our proof is much more involved for several reasons:
\begin{itemize}
\item we need to take care of the weaker structure of Galois rings and modules over them, e.g., zero divisors and the fact that not all modules have bases and thus module elements may not be uniquely represented in a minimal generating set;
\item we correct a few (rather minor) technical inaccuracies in the original proof; and
\item some for finite fields well-known prerequisite results are, to the best of our knowledge, not known over Galois rings.
\end{itemize}

Before analyzing the three conditions, we show the following result, whose implication is that if $\e$ is chosen randomly as described above, then the random variable $\Espace$, the support of the chosen error, is also uniformly distributed on the set of modules with rank profile $\phi$.
Note that the analogous statement for errors over a finite field follows immediately from linear algebra, but here, we need a bit more work.

\begin{lemma}\label{lem:uniform_support}
 Let $\phi(x) \in \mathbb Z[x]/(x^r)$ with nonnegative coefficients and let $\Espace$ be an $\Rq$-submodule of $\Rqm$ with rank profile $\phi(x)$. Then, the number of vectors $\e\in\Rqm^n$ whose support is equal to $\Espace$ only depends on $\phi(x)$.
\end{lemma}

\begin{proof}
 Let us write $\phi(x)=\sum_{i=0}^{r-1}n_ix^i$ with $N:=\phi(1)=\sum_{i=0}^{r-1}n_i=\rk_{\Rq}(\Espace)$, and let $\Gamma$ be a \mingenset for $\Espace$. Then, the vector $\e$ whose first $N$ entries are the element of $\Gamma$ and whose last $n-N$ entries are $0$ is a vector whose support is equal to $\Espace$. Moreover, all the vectors in $\Rqm^n$ whose support is equal to $\Espace$ are of the form $(\A\e^\top)^\top$, for $\A\in\GL(n,\Rq)$.
 Let us fix a basis of $\Rqm$ so that we can identify $\Rqm$ with $\Rq^m$. In this representation, $\e^\top$ corresponds to a matrix
 $\D\T$, where 
 \begin{equation*}
  \D=\begin{bmatrix} \I_{n_0} & & & &\\
  & g_\maxIdeal \I_{n_1} & & &\\
  & & \ddots & &\\
  & & & g_{\maxIdeal}^{r-1}\I_{n_{r-1}}&\\
  & & & & \0
\end{bmatrix}\in\Rq^{n\times n}
 \end{equation*}
  and $\T\in\Rq^{n\times m}$ has linearly independent rows over $\Rq$.
 Then, the vectors in $\Rqm^n$ whose support is equal to $\Espace$ correspond to matrices $\A\D\T$ for $\A\in \GL(n,\Rq)$, and their number is equal to the cardinality of the set
 \begin{equation*}
   \mathrm{Vec}(\Espace,n):=\{\A\D\T \mid \A\in\GL(n,\Rq)\}.
 \end{equation*}
 The group $\GL(n,\Rq)$ left acts on $\mathrm{Vec}(\Espace,n)$ and, by definition, its action is transitive. Hence, by the orbit-stabilizer theorem, we have
 \begin{equation*}
   |\mathrm{Vec}(\Espace,n)|=\frac{|\GL(n,\Rq)|}{|\mathrm{Stab}(\D\T)|},
 \end{equation*}
 where $\mathrm{Stab}(\D\T)=\mathrm{Stab}_{\GL(n,\Rq)}(\D\T)=\{\A\in \GL(n,\Rq) \mid \A\D\T=\D\T\}$.
 Hence, we need to count how many matrices $\A\in \GL(n,\Rq)$ satisfy
 \begin{equation*}
   (\A-\I_n)\D\T=0.
 \end{equation*}
 Let us call $\S:=\A-\I_n$ and divide it in $r+1$ block $\S_i\in\Rq^{n\times n_i}$ for $i\in\{0,\ldots,r-1\}$ and  $\S_r\in\Rq^{n\times (n-N)}$. Moreover, do the same with $\T$, dividing it in $r+1$ blocks $\T_i\in\Rq^{n_i\times m}$ for $i\in\{0,\ldots,r-1\}$ and  $\T_r\in\Rq^{(n-N)\times m}$. Therefore, we get
 \begin{equation*}
  \begin{bmatrix}\S_0 & \S_1 & \cdots & \S_{r-1} & \S_r\end{bmatrix}\begin{bmatrix}\T_0\\ g_{\maxIdeal} \T_1 \\ \vdots\\ g_{\maxIdeal}^{r-1}\T_{r-1}\\ \0\end{bmatrix}=\0. 
 \end{equation*}
 Since the rows of $\T$ are linearly independent over $\Rq$, this is true if and only if
 $\S_i\in\maxIdeal^{r-i}\Rq^{n\times n_i}$. This condition clearly only depends on the values $n_i$'s, and hence on $\phi(x)$.
\qed
\end{proof}

\subsection{Failure of Product Condition}\label{ssec:failure_product}

The product condition means that the product space of the randomly chosen support $\Espace$ and the fixed free module $\Fspace$ (in which the parity-check matrix coefficients are contained) has maximal rank profile $\phi^{\Espace \cdot \Fspace} = \phi^{\Espace} \phi^{\Fspace}$.
If $\Espace$ was a free module, the condition would translate to $\Espace \cdot \Fspace$ being a free module of rank $\lambda t$.
In fact, our proof strategy reduces the question if $\phi^{\Espace \cdot \Fspace} = \phi^{\Espace} \phi^{\Fspace}$ to the question whether a free module of rank $t$, which is related to $\Espace$, results in a product space with the free module $\Fspace$ of maximal rank profile.
Hence, we first study this question for products of free modules.
This part of the bound derivation is similar to the case of LRPC codes over finite fields (cf.~\cite{aragon2019low}), but the proofs and counting arguments are more involved since we need to take care of non-units in the ring.

\begin{lemma}\label{lem:product_space_full_dimension_probability_induction}
Let $\alpha',\beta$ be non-negative integers with $(\alpha'+1)\beta < m$.
Further, let $\Aspace',\Bspace$ be free submodules of $\Rqm$ of free rank $\alpha'$ and $\beta$, respectively, such that also $\Aspace'\cdot \Bspace$ is a free submodule of $\Rqm$ of free rank $\alpha'\beta$. For an element $a \in \Rqm^\ast$, chosen uniformly at random, let $\Aspace := \Aspace' + \langle a \rangle$. Then, we have
\begin{align*}
\Pr\big( \frk_{\Rq}(\Aspace\cdot\Bspace) < \alpha'\beta+\beta\big)
\leq \left(1-p^{-s\beta}\right) \sum_{j = 0}^{r-1} p^{s(r-j)[(\alpha'+1) \beta-m]}.
\end{align*}
\end{lemma}

\begin{proof}
First note that since $a$ is a unit in $\Rqm$, the mapping
$\varphi_a \, : \, \Bspace \to \Rqm, ~
b \mapsto ab$
is injective. This means that $a\Bspace$ is a free module with $\frk_{\Rq}(a\Bspace)=\frk_{\Rq}(\Bspace)=\beta$. Let $b_1,\dots,b_\beta$ be a basis of $\Bspace$.
Then, $a b_1, \dots, a b_\beta$ is a basis of $a\Bspace$.
Therefore, $\Aspace\cdot\Bspace$ is a free module with $\frk_{\Rq}(\Aspace\cdot\Bspace) = \alpha \beta+\beta$ if and only if $a\Bspace \cap \Aspace'\cdot\Bspace = \{0\}$.
Hence,
\begin{equation}
\Pr\big( \frk_{\Rq}(\Aspace\cdot\Bspace) < \alpha'\beta+\beta\big)
\leq \Pr\left( \exists b \in \Bspace  \setminus \{0\} : ab \in \Aspace'\cdot\Bspace \right). \label{eq:product_inequality_1}
\end{equation}
Let $c$ be chosen uniformly at random from $\Rqm$. Recall that $a$ is chosen uniformly at random from $\Rqm^*$. Then,
\begin{equation}
\Pr \! \left(  \exists b \in \Bspace \setminus \{0\} :  ab  \in  \Aspace' \cdot\Bspace   \right)
 \leq  \Pr \! \left( \exists b  \in  \Bspace  \setminus  \{0\}  :  cb  \in  \Aspace'\cdot\Bspace  \right). \label{eq:product_inequality_2}
\end{equation}
This holds since if $c$ is chosen to be a non-unit in $\Rqm$, then the statement ``$\exists \, b \in \Bspace \setminus \{0\} \, : \, cb \in \Aspace'\cdot\Bspace$'' is always true.
To see this, write $c = g_\maxIdeal c'$ for some $c' \in \Rqm$. Since $\beta>0$, there is a unit $b^* \in \Bspace \cap \Rqm^*$. 
Choose $b := g_\maxIdeal^{r-1}b^* \in \Bspace\setminus \{0\}$. Hence, $c b = g_\maxIdeal c' g_\maxIdeal^{r-1}b^* = 0$, and $b$ is from $\Bspace$ and non-zero.

Now we bound the right-hand side of \eqref{eq:product_inequality_2} as follows
\begin{align*}
\Pr\left( \exists  b \in  \Bspace  \setminus  \{0\}  :  cb \in  \Aspace'\cdot\Bspace \right)  &\leq \textstyle \sum_{b \in \Bspace \setminus \{0\}}  \Pr \left( cb \in \Aspace'\cdot\Bspace \right) \\
&=  \sum_{j = 0}^{r-1}  \sum_{b \in \Bspace :  v(b) = j}  \Pr\left( cb^* g_\maxIdeal^{j}  \in  \Aspace'\cdot\Bspace \right).
\end{align*}
Since $b^*$ is a unit in $\Rqm$, for uniformly drawn $c$, $c b^*$ is also uniformly distributed on $\Rqm$. Hence, $cb^* g_\maxIdeal^{j}$ is uniformly distributed on the ideal $\MaxIdeal^{j}$ of $\Rqm$ (the mapping $\Rqm \to \MaxIdeal^j$, $\chi \mapsto \chi g_\maxIdeal^j$ is surjective and maps equally many elements to the same image)
and we have
$\Pr\left( cb^* g_\maxIdeal^{j} \in \Aspace'\cdot\Bspace \right) = \frac{\left| \MaxIdeal^{j} \cap \Aspace'\cdot\Bspace \right|}{|\MaxIdeal^{j}|}$.
Let $v_1,\dots,v_{\alpha'\beta}$ be a basis of $\Aspace'\cdot\Bspace$. Then, by \eqref{eq:valuation}, an element $c \in \Aspace'\cdot\Bspace$ is in $\MaxIdeal^{j}$ if and only if it can be written as
$c = \sum_{i} \mu_i v_i$,
where $\mu_i \in \maxIdeal^j$ for all $i$.

Hence, 
$\left| \MaxIdeal^{j} \cap \Aspace'\cdot\Bspace \right| = |\maxIdeal^{j}|^{\alpha' \beta}$.
Moreover, we have
$|\MaxIdeal^{j}| = |\maxIdeal^{j}|^m$,
where
$|\maxIdeal^{j}| = p^{s(r-j)}$.
Overall, we get
\begin{align}
\Pr\left( \exists \, b \in  \Bspace  \setminus \{0\} \, : \, cb  \in  \Aspace'\cdot\Bspace \right) &\leq \sum_{j = 0}^{r-1} \sum_{b \in \Bspace \, : \, v(b) = j} p^{s(r-j)(\alpha' \beta-m)} \notag \\ 
&= \sum_{j = 0}^{r-1} \big|\{b \in \Bspace \, : \, v(b) = j\}\big| p^{s(r-j)(\alpha' \beta-m)}. \label{eq:product_inequality_3}
\end{align}
Furthermore, we have (note that $\MaxIdeal^{j+1} \subseteq \MaxIdeal^{j}$)
\begin{align}
\big|\{b \in \Bspace \, : \, v(b) = j\}\big| &= \Big|\big(\MaxIdeal^{j} \setminus \MaxIdeal^{j+1}\big) \cap \Bspace \Big| = \big|\MaxIdeal^{j} \cap \Bspace \big| - \big|\MaxIdeal^{j+1} \cap \Bspace \big| \notag \\
&= p^{s(r-j)\beta}-p^{s(r-j-1)\beta}. \label{eq:product_inequality_4}
\end{align}
Combining and simplifying \eqref{eq:product_inequality_1}, \eqref{eq:product_inequality_2}, \eqref{eq:product_inequality_3}, and \eqref{eq:product_inequality_4} we obtain the desired result.
\qed
\end{proof}

\begin{lemma}\label{lem:product_space_full_dimension_probability}
Let $\Bspace$ be a fixed free submodule of $\Rqm$ with $\frk_{\Rq}(\Bspace)=\beta$. For a positive integer $\alpha$ with $\alpha \beta <m$, let $\Aspace$ be drawn uniformly at random from the set of free submodules of $\Rqm$ of free rank $\alpha$.
Then,
\vspace{-0.005cm}
\begin{align*}
\Pr\left( \frk_{\Rq}(\Aspace\cdot\Bspace)< \alpha \beta \right)
\leq \left(1-p^{-s\beta}\right) \sum_{i=1}^{\alpha} \sum_{j = 0}^{r-1} p^{s(r-j)(i \beta-m)} \leq 2 \alpha p^{s(\alpha \beta-m)}
\end{align*}
\end{lemma}
\vspace{-0.15cm}
\begin{proof}
Drawing a free submodule $\Aspace \subseteq \Rqm$ of rank $\alpha$ uniformly at random is equivalent to drawing iteratively
$\Aspace_0 := \{0\}, ~
\Aspace_i := \Aspace_{i-1} + \langle a_i \rangle$ for $i=1,\dots,\alpha$
where for each iteration $i$, the element $a_i \in \Rqm$ is chosen uniformly at random from the set of vectors that are linearly independent of $\Aspace_{i-1}$. The equivalence of the two random experiments is clear since the possible choices of the sequence $a_1,\dots,a_\alpha$ gives exactly all bases of free $\Rq$-submodules of $\Rqm$ of rank $\alpha$. Furthermore, all sequences are equally likely and each resulting submodule has the same number of bases that generate it (which equals the number of invertible $\alpha \times \alpha$ matrices over $\Rq$).
We have the following recursive formula for any $i=1,\dots,\alpha$:
\begin{align*}
&\Pr\big( \frk_{\Rq}(\Aspace_i \cdot\Bspace)< i \beta \big) \\
&= \Pr\big( \frk_{\Rq}(\Aspace_i \cdot\Bspace)< i \beta \land \frk_{\Rq}(\Aspace
_{i-1}\cdot\Bspace)=(i-1)\beta\big) \\
&\quad \quad + \underbrace{\Pr\big( \frk_{\Rq}(\Aspace_i \cdot\Bspace)< i \beta \land \frk_{\Rq}(\Aspace_{i-1}\cdot\Bspace)<(i-1)\beta\big)}_{\text{$\frk_{\Rq}(\Aspace_{i-1}\cdot\Bspace)<(i-1)\beta$ implies $\frk_{\Rq}(\Aspace_i \cdot\Bspace)< i \beta$}} \\
&= \Pr\big( \frk_{\Rq}(\Aspace_i \cdot\Bspace)< i \beta  \mid \frk_{\Rq}(\Aspace_{i-1}\cdot\Bspace)=(i-1)\beta\big) \\
&\quad \quad\cdot \underbrace{\Pr(\frk_{\Rq}(\Aspace_{i-1}\cdot\Bspace)=(i-1)\beta)}_{\leq 1} + \Pr\big(\frk_{\Rq}(\Aspace_{i-1}\cdot\Bspace)<(i-1)\beta\big) \\
&\overset{(\ast)}{\leq} \left(1-p^{-s\beta}\right) \sum_{j = 0}^{r-1} p^{s(r-j)(i \beta-m)} + \Pr\big(\frk_{\Rq}(\Aspace_{i-1}\cdot\Bspace)<(i-1)\beta\big),
\end{align*}
where ($\ast$) follows from Lemma~\ref{lem:product_space_full_dimension_probability_induction} by the following additional argument:
\begin{align*}
&\Pr\big( \frk_{\Rq}(\Aspace_i \cdot\Bspace)< i \beta  \mid \frk_{\Rq}(\Aspace_{i-1}\cdot\Bspace)=(i-1)\beta\, \land\, a_i \text{ linearly independent and}\\
&\quad \quad \text{its span trivially intersects with $\Aspace_{i-1}$}\big) \\
&\leq \Pr\big( \frk_{\Rq}(\Aspace_i \cdot\Bspace)< i \beta  \mid \frk_{\Rq}(\Aspace_{i-1}\cdot\Bspace)=(i-1)\beta \, \land \, a_i \text{ uniformly from } \Rqm^* \big) \\
&\leq \left(1-p^{-s\beta}\right) \sum_{j = 0}^{r-1} p^{s(r-j)(i \beta-m)},
\end{align*}
where the last inequality is exactly the statement of Lemma~\ref{lem:product_space_full_dimension_probability_induction}.
By 
$\Pr\big(\frk_{\Rq}(A_{0}B)<0\big) = 0$, we get
\begin{align*}
\Pr\left( \frk_{\Rq}(\Aspace\cdot\Bspace)< \alpha \beta \right) &= \Pr\big( \frk_{\Rq}(\Aspace_\alpha \cdot\Bspace)< \alpha \beta \big) \\
&= \left(1-p^{-s\beta}\right) \sum_{i=1}^{\alpha} \sum_{j = 0}^{r-1} p^{s(r-j)(i \beta-m)} \\
&\leq \alpha \underbrace{\left(1-p^{-s\beta}\right)}_{\leq 1} p^{-rs(m-\alpha \beta)} \underbrace{\sum_{j = 0}^{r-1} p^{js(m-\alpha \beta)}}_{\leq 2 p^{(r-1)s(m-\alpha \beta)}} \\
&\leq 2 \alpha p^{s(\alpha \beta-m)}.
\end{align*}
This proves the claim.
\qed
\end{proof}

Recall that the error support $\Espace$ is not necessarily a free module.
In the following sequence of statements, we will therefore answer the question how the results of \new{Lemma~\ref{lem:product_space_full_dimension_probability_induction} and Lemma~\ref{lem:product_space_full_dimension_probability}} can be used to derive a bound on the product condition failure probability.
To achieve this, we study the following free modules related to modules of arbitrary rank profile.
Note that this part of the proof differs significantly from LRPC codes over finite fields, where all modules are vector spaces, and thus free.

For a module $\Mspace \subseteq \Rqm$ with \mingenset $\Gamma$, define $\FreeModule{\Gamma} \subseteq \Rqm$ be the free module that is obtained from $\Mspace$ as follows: Let us write $\Gamma=\{g_\maxIdeal^ia_{i,\ell_i} \mid 0\leq i \leq r-1, 1 \leq \ell_i \leq \phi^{\Mspace}_i \}$, where the elements $a_{i,\ell_i}$ are all reduced modulo $\MaxIdeal^{r-i}$, that is, the Teichm\"uller representation of $a_{i,\ell_i}$ is of the form
$$ a_{i,\ell_i}=\sum_{j=0}^{r-i-1} g_\maxIdeal^jz_j, \quad z_j\in T_{tm}.$$
This is clearly possible since if we add to $a_{i,\ell_i}$ an element $y\in \MaxIdeal^{r-i}=(g_\maxIdeal^{r-i})$, then  $g_\maxIdeal^i(a_{i,\ell_i}+y)=g_\maxIdeal^ia_{i,\ell_i}$. At this point, we define  $\FreeModuleBasis{\Gamma} := \{a_{i,\ell_i} \mid 0\leq i \leq r-1, 1 \leq \ell_i \leq \phi^{\Mspace}_i\}$, and $\FreeModule{\Gamma}:=\langle \FreeModuleBasis{\Gamma}\rangle_{\Rq}$. The fact that $\FreeModule{\Gamma}$ is free directly follows from considering its Smith Normal Form, which tells us that in the matrix representation it is spanned by (some of) the rows of an invertible matrix in $\GL(m,\Rq)$. In particular, we have $\frk_{\Rq}(\FreeModule{\Gamma})=\rank_{\Rq}(\Mspace)$.

\begin{example}
 Let $p=2$, $s=1$, $r=3$ as in Example~\ref{exp:Rq}, $h(z) = z^3+z+1$ and $\Mspace$ a module with \mingenset $\Gamma = \{1,2z^2+2z,4z^2+2z+2\}$. Then, $\Mspace$ has a diagnonal matrix in Smith normal form of
  \begin{equation*}
    \begin{bmatrix}
      1&0&0\\
      0&2&0\\
      0&0&2
    \end{bmatrix}
  \end{equation*}
  and $\phi^{\Mspace}(z) = 2z+1$. Using the notation above, we observe $a_{0,1}=1$, $a_{1,1}=z^2+z$, $a_{1,2} = z^3+2z^2$ and $\FreeModule{\Gamma} = \langle \{1,z^2+z,z^3+2z^2\}  \rangle_{\Rq}$. 
\end{example}

At this point, for two different \mingensets $\Gamma, \Lambda$ of $\Mspace$, one could ask whether  $\FreeModule{\Gamma}= \FreeModule{\Lambda}$. The answer is affirmative, and it can be deduced from the following result.

\begin{proposition}\label{prop:freemodule_matrix}
 Let $n_0,\ldots,n_{r-1}\in \mathbb{N}$ be nonnegative integers, let $N:=n_0+\ldots+n_{r-1}$ and let $\D\in \Rq^{N\times N}$ be a diagonal matrix given by 
 \begin{equation*}
   \D:=\begin{bmatrix} \I_{n_0} & & & \\
  & g_\maxIdeal \I_{n_1} & & \\
  & & \ddots & \\
  & & & g_{\maxIdeal}^{r-1}\I_{n_{r-1}}
 \end{bmatrix}.
 \end{equation*}
 Moreover, let $\T_1,\T_2 \in \Rq^{r\times m}$ be such that the rows of $\T_i$ are $\Rq$-linearly independent for each $i\in\{1,2\}$. Then, the rowspaces of $\D\T_1$ and $\D\T_2$ coincide 
 if and only if for every $i,j \in \{0,\ldots,r-1\}$ there exist $\Y_{i,j}\in\Rq^{n_i\times n_j}$ with $\Y_{i,i}\in \GL(n_i,\Rq)$ and $\Z_i\in\Rq^{n_i\times m}$ such that $$\T_2=\Y\T_1+\Z,$$ where 
 \begin{equation*}
   \Y= \begin{bmatrix} \Y_{0,0} & g_\maxIdeal \Y_{0,1} & g_\maxIdeal^2 \Y_{0,2} & \cdots & g_\maxIdeal^{r-1} \Y_{0,r-1} \\ 
 \Y_{1,0} & \Y_{1,1} & g_\maxIdeal \Y_{1,2} & \cdots & g_\maxIdeal^{r-2} \Y_{1,r-1} \\
 \Y_{2,0} & \Y_{2,1} &  \Y_{2,2} & \cdots & g_\maxIdeal^{r-3} \Y_{2,r-1} \\
 \vdots & \vdots & \vdots & &\vdots \\
 \Y_{r-1,0} & \Y_{r-1,1} & \Y_{r-1,2} & \cdots & \Y_{r-1,r-1} \\
 \end{bmatrix}, \quad \Z=\begin{bmatrix} \0 \\
 g_\maxIdeal^{r-1}\Z_1 \\ g_\maxIdeal^{r-2}\Z_2 \\ \vdots \\ g_\maxIdeal \Z_{r-1}\end{bmatrix}. 
 \end{equation*}
 
\end{proposition}

\begin{proof}
 The rowspaces of $\D\T_1$ and $\D\T_2$ coincide if and only if there exists a matrix $\XX\in \GL(N,\Rq)$ such that $\XX\D\T_1=\D\T_2$. Divide $\T_\ell$ in $r$ blocks $\T_{\ell,i}\in \Rq^{n_i \times m}$ for $i\in\{0,\ldots, r-1\}$ and divide $\XX$ in $r\times r$ blocks $\XX_{i,j}\in \Rq^{n_i\times n_j}$ for $i,j \in\{0,\ldots,r-1\}$. Hence, from $\XX\D\T_1=\D\T_2$ we get 
 \begin{equation}\label{eq:T_1T_2}
 \sum_{j=0}^{r-1} \XX_{i,j}g_\maxIdeal^j\T_{1,j}=g_\maxIdeal^i\T_{2,i}.
 \end{equation}
 Since the rows of $\T_{1}$ are $\Rq$-linearly independent,  \eqref{eq:T_1T_2} implies that $g_\maxIdeal^j\XX_{i,j} \in g_\maxIdeal^i\Rq^{n_i\times n_j}$. This shows that 
 \begin{equation*}
   \XX= \begin{bmatrix} \Y_{0,0} &  \Y_{0,1} & \Y_{0,2} & \cdots &  \Y_{0,r-1} \\ 
 g_\maxIdeal \Y_{1,0} & \Y_{1,1} &  \Y_{1,2} & \cdots &  \Y_{1,r-1} \\
 g_\maxIdeal^2  \Y_{2,0} & g_\maxIdeal \Y_{2,1} &  \Y_{2,2} & \cdots &  \Y_{2,r-1} \\
 \vdots & \vdots & \vdots & &\vdots \\
 g_\maxIdeal^{r-1} \Y_{r-1,0} & g_\maxIdeal^{r-2} \Y_{r-1,1} & g_\maxIdeal^{r-3} \Y_{r-1,2} & \cdots & \Y_{r-1,r-1} \\
 \end{bmatrix},
 \end{equation*}
 for some $\Y_{i,j}\in \Rq^{n_i\times n_j}$. Observe now that $\XX=\U+g_\maxIdeal\LL$, where 
\begin{align*}
\U &= \begin{bmatrix} \Y_{0,0} &  \Y_{0,1} & \Y_{0,2} & \cdots &  \Y_{0,r-1} \\ 
\0 & \Y_{1,1} &  \Y_{1,2} & \cdots &  \Y_{1,r-1} \\
\0 & \0 &  \Y_{2,2} & \cdots &  \Y_{2,r-1} \\
\vdots & \vdots & \vdots & &\vdots \\
\0 & \0 & \0 & \cdots & \Y_{r-1,r-1} \\
\end{bmatrix}, \\
\LL &=\begin{bmatrix} \0 &  \0 & \0 & \cdots &  \0 \\ 
\Y_{1,0} & \0 &  \0 & \cdots & \0 \\
g_\maxIdeal \Y_{2,0} &  \Y_{2,1} &  \0 & \cdots &  \0 \\
\vdots & \vdots & \vdots & &\vdots \\
g_\maxIdeal^{r-2} \Y_{r-1,0} & g_\maxIdeal^{r-3} \Y_{r-1,1} & g_\maxIdeal^{r-4} \Y_{r-1,2} & \cdots & \0 \\
\end{bmatrix}.
\end{align*}

 Since $\XX$ is invertible and $g_\maxIdeal\LL$ is nilpotent, then $\U$ is also invertible and hence  $\Y_{i,i}\in\GL(n_i,\Rq)$, for every $i\in\{0,\ldots,r-1\}$. At this point, observe that $ \XX\D= \D\Y$, from which we deduce
 $$ \D(\T_2-\Y\T_1)=\0.$$
 This implies that the $i$th block of $\T_2-\Y\T_1 \in \Ann(g_\maxIdeal^i)\Rq^{n_i\times m}=g_\maxIdeal^{r-i}\Rq^{n_i\times m}$ and we conclude.
\qed
\end{proof}

Let $\Mspace$ be an $\Rq$-submodule of $\Rqm$. Proposition \ref{prop:freemodule_matrix} implies that if we restrict to take  a \mingenset $\Gamma=\{g_\maxIdeal^ia_{i,j_i} \mid 0\leq i \leq r-1, 1 \leq j_i \leq \phi^{\Mspace}_i\}$ such that the  elements $a_{i,j_i}$ have Teichm\"uller representation 
\begin{equation}\label{eq:Teichmuller_mingensets}
a_{i,j_i}=\sum_{\ell=0}^{r-i-1}g_\maxIdeal^\ell z_\ell, \quad z_\ell\in T_{tm},
\end{equation}
then the module $\FreeModule{\Gamma}$ is well-defined and does not depend on the choice of $\Gamma$.

\begin{definition}\label{def:FreeModule}
 We define $\FreeModule{\Mspace}$ to be the space $\FreeModule{\Gamma}$, where $\Gamma=\{g_\maxIdeal^ia_{i,j_i} \mid 0\leq i \leq r-1, 1 \leq j_i \leq \phi^{\Mspace}_i\}$ is any \mingenset such that the elements  $a_{i,j_i}$ have Teichm\"uller representation  as in \eqref{eq:Teichmuller_mingensets}.
\end{definition}

The following two corollaries follow from observations in Proposition~\ref{prop:freemodule_matrix}. We will use them to show that for certain uniformly chosen modules $\Mspace$, the corresponding free modules $\FreeModule{\Mspace}$ are uniformly chosen from the set of free modules of rank equal to the rank of $\Mspace$. The proofs can be found in Appendix~\ref{app:proofs_freemodule_corollaries}.

Now, for a given $\Rq$-submodule of $\Rqm$ we consider all the free modules that comes from a \mingenset for $\Mspace$. More specifically, we set
\begin{align*}
\mathrm{Free}(\Mspace):=\Big\{ \Aspace \mid &\Aspace \mbox{ is free with } \frk_{\Rq}(\Aspace)=\rk_{\Rq}(\Mspace) \mbox{ and } \exists \{ a_{i,\ell_i}\} \mbox{  basis of } \Aspace \\ & \mbox{ such that } \{ g_\maxIdeal^ia_{i,\ell_i}\} \mbox{ is a \mingenset for } \Mspace   \Big\}.
\end{align*}
In fact, even though for the $\Rq$-module $\Mspace$ there is a unique free module $\FreeModule{\Mspace}$ as explained in Definition \ref{def:FreeModule}, we have more than one free module $\Aspace$ belonging to $\mathrm{Free}(\Mspace)$. The exact number of such free modules is given in the following Corollary.

\begin{corollary}\label{cor:num_free}
 Let $\Mspace$ be an $\Rq$-submodule of $\Rqm$ with rank profile $\phi^{\Mspace}(x)$ and rank $N := \rank_{\Rq}(\Mspace)$. Then $$|\mathrm{Free}(\Mspace)|=s^{(m-N)\sum_{i=1}^{r-1}i \phi^{\Mspace}_i }.$$
 In particular, $|\mathrm{Free}(\Mspace)|$ only depends on $\phi^{\Mspace}(x)$.
\end{corollary}

\begin{proof}
  See Appendix~\ref{app:proofs_freemodule_corollaries}.
 \qed 
\end{proof}

Now we estimate an opposite quantity.  For a fixed rank profile $\phi(x)$ with $\phi(1)\leq m$, and given a free $\Rq$-submodule $\Nspace$ of $\Rqm$ with free rank $\frk_{\Rq}(\Nspace)=\phi(1)$, for how many $\Rq$-submodules $\Mspace$ of $\Rqm$ with rank profile $\phi^{\Mspace}(x)=\phi(x)$ the module $\Nspace$ belongs to $\mathrm{Free}(\Mspace)$?
Formally, we want to estimate the cardinality of the set
$$ \mathrm{Mod}(\phi,\Nspace):=\left\{\Mspace \subseteq \Rqm \mid \phi^{\Mspace}(x)=\phi(x) \mbox{ and } \Nspace \in \mathrm{Free}(\Mspace) \right\}.$$

\begin{corollary}\label{cor:num_mod}
  Let $\phi(x)=\sum_{i=0}^{r-1}n_ix^i\in \mathbb{N}[x]/(x^r)$ such that $\phi(1)=N\leq m$, and let $\Nspace$ be a free $\Rq$-submodule of $\Rqm$ with free rank $\frk_{\Rq}(\Nspace)=N$. Then $$|\mathrm{Mod}(\phi,\Nspace)|=\frac{|\GL(N,\Rq)|}{|G_{\phi}^*|}.$$
  In particular, $|\mathrm{Mod}(\phi,\Nspace)|$ only depends on $\phi(x)$.
\end{corollary}

\begin{proof}
See Appendix~\ref{app:proofs_freemodule_corollaries}.
\qed
\end{proof}

We need the following lemma to derive a sufficient condition for the product of two modules to have a maximal rank profile. 

\begin{lemma}\label{lem:product_free_differentbases}
Let $\Mspace$ be an $\Rq$-submodule  of $\Rqm$, and let $\Aspace, \Bspace \in \mathrm{Free}(\Mspace)$. Moreover, let $\Nspace$ be a free $\Rq$-submodule of $\Rqm$.  
Then, $\Nspace\cdot \Aspace$ is free with   $\frk_{\Rq}(\Nspace\cdot \Aspace)=\rk_{\Rq}(\Mspace)\frk_{\Rq}(\Nspace)$  if and only if $\Nspace\cdot \Bspace$ is free with $\frk_{\Rq}(\Nspace\cdot \Bspace)=\rk_{\Rq}(\Mspace)\frk_{\Rq}(\Nspace)$.
\end{lemma}

\begin{proof}
 Let $A=\{a_{i,j_i} \mid 0\leq i \leq r-1, 1 \leq j_i \leq \phi^{\Mspace}_i\}$ be a basis of $\Aspace$ and  $B=\{b_{i,j_i} \mid 0\leq i \leq r-1, 1 \leq j_i \leq \phi^{\Mspace}_i\}$ be a basis of $\Bspace$ such that  $\Gamma:= \{g_\maxIdeal^ia_{i,j_i} \mid 0\leq i \leq r-1, 1 \leq j_i \leq \phi^{\Mspace}_i\}$ and $\Lambda:= \{g_\maxIdeal^ib_{i,j_i} \mid 0\leq i \leq r-1, 1 \leq j_i \leq \phi^{\Mspace}_i\}$ are two \mingensets for $\Mspace$, and let $\Delta=\{u_1,\ldots,u_t\}$ be a basis for $\Nspace$. Assume that $\Delta\cdot A=\{u_{\ell}a_{i,j_i} \}$ has $\rk_{\Rq}(\Mspace)\frk_{\Rq}(\Nspace)$ linearly independent elements over $\Rq$. By symmetry, it is enough to show that this implies $\Nspace\cdot \Bspace$ is free. 
 By Proposition \ref{prop:freemodule_matrix}, we know that there exists $x_{i,j_i} \in \Rqm$ such that $\Bspace=\langle \{a_{i,j_i}+g_\maxIdeal x_{i,j_i} \mid 0\leq i \leq r-1, 1 \leq j_i \leq \phi^{\Mspace}_i \} \rangle_{\Rq}$. Hence, we need to prove that the elements $\{u_{\ell}(a_{i,j_i}+g_mx_{i,j_i})\}$ are linearly independent over  $\Rq$.  Suppose that there exists $\lambda_{\ell,i,j_i}\in\Rq$ such that
 $$ \sum_{\ell,i,j_i}\lambda_{\ell,i,j_i}u_{\ell}(a_{i,j_i}+g_mx_{i,j_i})=0,$$
 hence, rearranging the sum, we get
 \begin{equation}\label{eq:ascending_equality}  \sum_{\ell,i,j_i}\lambda_{\ell,i,j_i}u_{\ell}a_{i,j_i}=- g_\maxIdeal\sum_{\ell,i,j_i}\lambda_{\ell,i,j_i}u_{\ell}x_{i,j_i}. \end{equation}
 Multiplying both sides by $g_\maxIdeal^{r-1}$ we obtain
 $$ \sum_{\ell,i,j_i}\lambda_{\ell,i,j_i}g_\maxIdeal^{r-1}u_{\ell}a_{i,j_i}=0,$$
 and since by hypothesis $\{u_{\ell}a_{i,j_i}\}$ is a basis, this implies 
 $\lambda_{\ell,i,j_i}\in \Ann(g_\maxIdeal^{r-1})=\maxIdeal$ and therefore there exist
 $\lambda'_{\ell,i,j_i}\in \Rq$, such that $\lambda_{\ell,i,j_i}=\g_\maxIdeal\lambda'_{\ell,i,j_i}$. Thus, \eqref{eq:ascending_equality} becomes
 $$g_\maxIdeal\sum_{\ell,i,j_i}\lambda'_{\ell,i,j_i}u_{\ell}a_{i,j_i}=- g_\maxIdeal^2\sum_{\ell,i,j_i}\lambda'_{\ell,i,j_i}u_{\ell}x_{i,j_i}.$$
 Now, multiplying both sides by $g_\maxIdeal^{r-2}$ and with the same reasoning as before, we obtain that all the $\lambda'_{\ell,i,j_i}\in\maxIdeal$ and the right-hand side of \eqref{eq:ascending_equality} belongs to $\maxIdeal^3$. Iterating this process $r-2$ times, we finally get that the right-hand side of \eqref{eq:ascending_equality} belongs to $\maxIdeal^r=(0)$, and therefore \eqref{eq:ascending_equality} corresponds to
 $$ \sum_{\ell,i,j_i}\lambda_{\ell,i,j_i}u_{\ell}a_{i,j_i}=0,$$
 which, by hypothesis implies $\lambda_{\ell,i,j_i}=0$ for every $\ell,i,j_i$. This concludes the proof, showing that the elements $\{u_{\ell}(a_{i,j_i}+g_mx_{i,j_i})\}$ are linearly independent over  $\Rq$.
\qed
\end{proof}

With the aid of Lemma~\ref{lem:product_free_differentbases} we can show that the property for the product of two arbitrary $\Rq$-modules $\Mspace_1, \Mspace_2$ of having maximal rank profile (according to Definition \ref{def:ordering_rankprofiles}) depends on the  free modules $\FreeModule{\Mspace_1}$ and $\FreeModule{\Mspace_2}$ and on their product.

\begin{proposition}\label{prop:maximal_product_space}
Let $\Mspace_1$ and $\Mspace_2$ be submodules of $\Rqm$.
If the product of free modules $\FreeModule{\Mspace_1}$ and $\FreeModule{\Mspace_2}$ has free rank 
\begin{equation*}
\frk_{\Rq}\!\left(\FreeModule{\Mspace_1}\FreeModule{\Mspace_2}\right) = \rank_{\Rq}(\FreeModule{\Mspace_1}) \rank_{\Rq}(\FreeModule{\Mspace_2}),
\end{equation*}
then we have
\begin{equation*}
\phi^{\Mspace_1 \cdot \Mspace_2}(x) = \phi^{\Mspace_1}(x) \phi^{\Mspace_2}(x).
\end{equation*}

Moreover, if we assume that $\deg(\phi^{\Mspace_1}(x))+\deg(\phi^{\Mspace_2}(x))<r$, then also the converse is true.
In particular, the converse is true if one of the two modules is free.
\end{proposition}

\begin{proof}
First, observe that by Lemma \ref{lem:product_free_differentbases} we can take any pair of \mingensets $\Gamma_1$ and $\Gamma_2$ of $\Mspace_1$ and $\Mspace_2$, respectively. Let us fix 
$$\Gamma_1:=\{g_\maxIdeal^ia_{i,j_i} \mid 0\leq i \leq r-1, 1 \leq j_i \leq \phi^{\Mspace_1}_i \}$$
 \mingenset of $\Mspace_1$ and
$$\Gamma_2:=\{g_\maxIdeal^ib_{i,j_i} \mid 0\leq i \leq r-1, 1 \leq j_i \leq \phi^{\Mspace_2}_i \}$$
\mingenset of $\Mspace_2$.
By hypothesis, the set
$\FreeModuleBasis{\Gamma_1}\cdot \FreeModuleBasis{\Gamma_2}$ contains $\rk_{\Rq}(\Mspace_1)\rk_{\Rq}(\Mspace_2)=t$ linearly independent elements over $\Rq$. Let $\A\in \Rq^{t\times m}$ be the matrix whose rows are the vectorial representations in $\Rq^m$ of the elements in $\FreeModuleBasis{\Gamma_1}\cdot \FreeModuleBasis{\Gamma_2}$. Clearly, a Smith Normal Form for $\A$ is $\A=\D\T$ where  $\D= ( \I_t \mid \0)$ and $\T\in \GL(n,\Rq)$ is any invertible matrix whose first $t\times m$ block is equal to $\A$. 
By definition $\Gamma_1\cdot \Gamma_2$ is a generating set for $\Mspace_1\cdot \Mspace_2$ and hence $\Mspace_1\cdot \Mspace_2$ is equal to the rowspace of the matrix $\A'$ whose rows are the vectorial representations of the elements in $\Gamma_1\cdot \Gamma_2$. A row of $\A'$  corresponding to the element $g_\maxIdeal^ia_{i,j_i}g_\maxIdeal^sb_{s,\ell_s}\in \Gamma_1\cdot\Gamma_2$ is equal to the row of $\A$ corresponding to the element $a_{i,j_i}b_{s,\ell_s}$ multiplied by $g_\maxIdeal^{i+s}$. Therefore, $\A'=\D'\A=\D'\D\T=(\D'\mid \0)\T$, where $\D'$ is a $t\times t$ diagonal matrix whose diagonal elements are all of the form $g_{\maxIdeal}^{i+s}$ for suitable $i,s$. This shows that $\A'=(\D'\mid \0)\T$ is a Smith Normal Form for $\A'$ and the rank profile  $\phi^{\Mspace_1\cdot \Mspace_2}(x)$ corresponds to $\phi^{\Mspace_1}(x)\phi^{\Mspace_2}(x)$.

On the other hand, if $\phi^{\Mspace_1\cdot \Mspace_2}(x)=\phi^{\Mspace_1}(x)\phi^{\Mspace_2}(x)$, then the set $\Gamma_1\cdot\Gamma_2$ is a \mingenset for $\Mspace_1\cdot\Mspace_2$. Moreover, since $\deg(\phi^{\Mspace_1}(x))+\deg(\phi^{\Mspace_2}(x))<r$, we have that $F(\Gamma_1)\cdot F(\Gamma_2)=F(\Gamma_1\cdot\Gamma_2)$,
which is a set of $\rk_{\Rq}(\Mspace_1)\rk_{\Rq}(\Mspace_2)$ nonzero elements. Let $\S\D\T$ be a Smith normal form for $\Mspace_1\cdot\Mspace_2$, then the elements of $F(\Gamma_1\cdot\Gamma_2)$ correspond to the first $\rk_{\Rq}(\Mspace_1)\rk_{\Rq}(\Mspace_2)$ rows of matrix $\T$, and hence they are $\Rq$-linearly independent. Thus, $\FreeModule{\Mspace_1}\cdot \FreeModule{\Mspace_2}$ is free with free rank equal to $\rk_{\Rq}(\Mspace_1)\rk_{\Rq}(\Mspace_2)$.
\qed
\end{proof}

\begin{remark}
Observe that the second part of Proposition \ref{prop:maximal_product_space} does not hold anymore if we remove the hypothesis that $\deg(\phi^{\Mspace_1}(x))+\deg(\phi^{\Mspace_2}(x))<r$.

  Let $\Aspace'$, $\Aspace=\Aspace'+\langle a\rangle$ and $\Bspace$ be three free modules of free rank $\alpha-1$, $\alpha$ and $\beta$ respectively, such that $\Aspace'\cdot \Bspace$ is free of rank $(\alpha-1)\beta$, but $\Aspace\cdot \Bspace$ is not free of rank $\alpha\beta$. Take a basis for $\Aspace$ of the form $\{a_1,\ldots, a_{\alpha-1},a\}$ such that $\{a_1,\ldots, a_{\alpha-1}\}$ is a basis of $\Aspace'$, and fix also a basis $\{b_1,\ldots,b_{\beta}\}$ for $\Bspace$. Then, define $\Mspace_1$ to be the $\Rq$-module whose \mingenset is $\{a_1,\ldots,\a_{\alpha-1},g_\maxIdeal^{r-1}a\}$, and define $\Mspace_2=\maxIdeal \Bspace$. Consider the module $\Mspace_1\cdot \Mspace_2$. It is easy to see that  $\Mspace_1\cdot \Mspace_2=\maxIdeal (\Aspace'\cdot\Bspace)=\Aspace'\cdot \Mspace_2$. Observe that $\Bspace\in\mathrm{Free}(\Mspace_2)$ and by Proposition \ref{prop:maximal_product_space} and Lemma \ref{lem:product_free_differentbases}, we have that $\phi^{\Mspace_1 \cdot \Mspace_2}(x) = \phi^{\Mspace_1}(x) \phi^{\Mspace_2}(x)$. However, by construction  we have $\Aspace\in\mathrm{Free}(\Mspace_1)$,  $\Bspace\in\mathrm{Free}(\Mspace_1)$ and $\Aspace \cdot \Bspace$ is not free of rank $\alpha \beta$. Therefore, by Lemma \ref{lem:product_free_differentbases} this also holds for $\FreeModule{\Mspace_1}\cdot \FreeModule{\Mspace_2}$.
\end{remark}

We are now ready to put the various statements of this subsection together and prove an upper bound on the failure probability of the product condition---the main statement of this subsection.

\begin{theorem}\label{thm:general_product}
Let $\Bspace$ be a fixed $\Rq$-submodule of $\Rqm$ with rank profile $\phi^{\Bspace}(x)$ and let $\lambda:=\phi^{\Bspace}(1)=\rk_{\Rq}(\Bspace)$.
Let $t$ be a positive integer with $t \lambda <m$ and $\phi(x) \in \ZZ[x]/(x^r)$ with nonnegative coefficients such that $\phi(1)=t$.
Let $\Aspace$ be an $\Rq$-submodule of $\Rqm$ selected uniformly at random among all the modules with  $\phi^\Aspace = \phi$.
Then,
\begin{align*}
\Pr\left( \phi^{\Aspace \cdot \Bspace} \neq \phi^{\Aspace} \phi^{\Bspace} \right) %
\leq \left(1-p^{-s\beta}\right) \sum_{i=1}^{\alpha} \sum_{j = 0}^{r-1} p^{s(r-j)(i \beta-m)} \leq 2 \alpha p^{s(\alpha \beta-m)}
\end{align*}
\end{theorem}

\begin{proof}
Let us denote by $\mathrm{Mod}(\phi)$ the set of all $\Rq$-submodules of $\Rqm$ whose rank profile equals $\phi$.
Choose uniformly at random a module $\Aspace$ in $\mathrm{Mod}(\phi)$, and then select $\X$ uniformly at random from $\mathrm{Free}(\Aspace)$. Then, this results in a uniform distribution on the set of all free modules with free rank equal to $\phi(1)=t$, that is the set $\mathrm{Mod}(t)$, where $t$ denotes the constant polynomial in $\ZZ[x]/(x^r)$ equal to $t$. Indeed, for an arbitrary free module $\Nspace$ with $\frk_{\Rq}(\Nspace)=t$,
\begin{align*}
\Pr(\X=\Nspace)&=\Pr(\X=\Nspace \mid \Aspace \in \mathrm{Mod}(\Nspace,\phi))\Pr(\Aspace \in \mathrm{Mod}(\Nspace,\phi))\\
&=\frac{1}{|\mathrm{Free}(\Aspace)|}\frac{|\mathrm{Mod}(\Nspace,\phi)|}{|\mathrm{Mod}(\phi)|},\end{align*}
which by Corollaries \ref{cor:num_free} and \ref{cor:num_mod} is a constant number that does not depend on $\Nspace$.

Now, suppose that $\phi^{\Aspace \cdot \Bspace} \neq \phi^{\Aspace} \phi^{\Bspace}$.
By Proposition \ref{prop:maximal_product_space}, this implies $\Nspace\cdot \Nspace'$ is not a free module of rank $t\lambda$, where $\Nspace$ is any free module in $\mathrm{Free}(\Aspace)$
and $\Nspace'$ is any  free module in $\mathrm{Free}(\Bspace)$. Hence,
\begin{align*}\Pr\left( \phi^{\Aspace \cdot \Bspace} \neq \phi^{\Aspace} \phi^{\Bspace} \right) 
\leq 1- \Pr\big(\text{$\Nspace\cdot\Nspace'$ is a free module of free rank $t\lambda$}\big),
\end{align*}
and we conclude using Lemma \ref{lem:product_space_full_dimension_probability}.
\qed
\end{proof}

As a consequence, we can finally derive the desired upper bound on the product condition failure probability.

\begin{theorem}\label{thm:product_condition_main_statement}
Let $\Fspace$ be defined as in Definition~\ref{def:LRPCcodes}.
Let $t$ be a positive integer with $t \lambda <m$ and $\phi(x) \in \ZZ[x]/(x^r)$ with nonnegative coefficients and such that $\phi(1)=t$ (recall that this means that an error of rank profile $\phi$ has rank $t$).
Let $\e$ be an error word, chosen uniformly at random among all error words with support $\Espace$ of rank profile $\phi^\Espace = \phi$.
Then, the probability that the {product condition} is not fulfilled is
\begin{align*}
&\Pr\left( \phi^{\Espace \cdot \Fspace} \neq \phi^{\Espace} \phi^{\Fspace} \right)
\leq \left(1-p^{-s\lambda}\right) \sum_{i=1}^{t} \sum_{j = 0}^{r-1} p^{s(r-j)(i \lambda-m)} \leq 2 t p^{s(t \lambda-m)}
\end{align*}
\end{theorem}

\begin{proof}
Let us denote by $\mathrm{Mod}(\phi)$ the set of all $\Rq$-submodules of $\Rqm$ whose rank profile equals $\phi$. By Lemma \ref{lem:uniform_support}, choosing uniformly at random $\e$ among all the words whose support $\Espace$ has rank profile $\phi$ results in a uniform distribution on $\mathrm{Mod}(\phi)$.  At this point, the claim follows from Theorem \ref{thm:general_product}.
\qed
\end{proof}

\subsection{Failure of Syndrome Condition}\label{ssec:failure_syndrome}

Here we derive a bound on the probability that the syndrome condition is not fulfilled, given that the product condition is satisfied.
As in the case of finite fields, the bound is based on the relative number of matrices of a given dimension that have full (free) rank.
\new{For completeness, we give a closed-form expression for this number in the following lemma. However, it can also be derived from the number of submodules of a given rank profile, which was given in \cite[Theorem~2.4]{honold2000linear}. Note that the latter result holds also for finite chain rings.}

\begin{lemma}\label{lem:syndrome_number_of_full-rank_matrices}
Let $a,b$ be positive integers with $a < b$.
Then, the number of $a \times b$ matrices over $\Rq=\GR(p^r,s)$ of (full) free rank $a$ is
$\NM(a,b;\Rq) = p^{a b r s} \prod_{a'=0}^{a-1} \left(1-p^{a'-b} \right)$.
\end{lemma}

\begin{proof}
First note that $\NM(1,b;\Rq) = p^{b r s}-p^{b (r-1) s} = p^{brs}\big(1-p^{bs}\big)$ since a $1 \times b$ matrices over $\Rq$ is of free rank $1$ if and only if at least one entry is a unit. Hence we subtract from the number of all matrices ($|\Rq|^b = p^{b r s}$) the number of vectors that consist only of non-units $(|\Rq|-|\Rq^*|)^b = p^{b(r-1)s}$ (cf.~\eqref{eq:Rq_number_units}).

Let now for any $a' \leq a$ be $\A \in \Rq^{a' \times b}$ a matrix of free rank~$a'$. We define 
$\mathcal{V}(\A) := \big\{ \v \in \Rq^{1 \times b} \! : \! \frk \big(\begin{bmatrix}
\A^\top 
\v^\top
\end{bmatrix}^\top \big) = a' \big\}$.
We study the cardinality of $\mathcal{V}(\A)$.
We have $\frk\big(\begin{bmatrix}
\A^\top
\v^\top
\end{bmatrix}^\top\big) = a'$ if and only if the rows of the matrix $\hat{\A} := \begin{bmatrix}
\A^\top
\v^\top
\end{bmatrix}^\top$ are linearly dependent.
Due to $\frk(\A) = a'$ and the existence of a Smith normal form of $\A$, there are invertibe matrices $\S$ and $\T$ such that
$\S \A \T = \D$,
where $\D$ is a diagonal matrix with ones on its diagonal.

Since $\S$ and $\T$ are invertible, we can count the number of vectors $\v'$ such that the rows of the matrix
$\big[
\D^\top
{\v'}^\top
\big]^\top$
are linearly independent instead of the matrix $\hat{\A}$ (note that $\v = \v' \T^{-1}$ gives a corresponding linearly dependent row in~$\hat{\A}$).

Since $\D$ is in diagonal form with only ones on its diagonal, the linearly dependent vectors are exactly of the form
\begin{equation*}
\v' = [v'_1,\dots,v'_a,v'_{a'+1},\dots, v'_b],
\end{equation*}
where $v'_i \in \Rq$ for $i=1,\dots,a'$ and $v'_i \in \maxIdeal$ for $i=a'+1,\dots,b$. Hence, we have 
\begin{equation*}
|\mathcal{V}(\A)| = p^{a'rs} p^{(b-a')(r-1)s)} = p^{brs} p^{(a'-b)s}.
\end{equation*}
Note that this value is independent of $\A$.

By the discussion on $|\mathcal{V}(\A)|$, we get the following recursive formula:
\begin{align*}
\NM(a'\!+ \!1,b;\Rq) \! = \! \begin{cases}
\NM(a',b;\Rq) p^{brs}\! \left( 1- p^{(a'-b)s} \right), \! \! &a'\geq 1, \\
p^{brs}\!\left(1-p^{bs}\right), &a'=0,
\end{cases}
\end{align*}
which resolves into
$\NM(a,b;\Rq) = p^{a b r s} \prod_{a'=0}^{a-1} \left(1-p^{(a'-b)s}\right)$.
\qed
\end{proof}

At this point we can prove the bound on the failure probability of the syndrome condition similar to the one in \cite{gaborit2013low}, using Lemma~\ref{lem:syndrome_number_of_full-rank_matrices}. The additional difficulty over rings is to deal with non-unique decompositions of module elements in \mingensets and the derivation of a simplified bound on the relative number of non-full-rank matrices.
Furthermore, the start of the proof corrects a minor technical impreciseness of Gaborit et al.'s proof.

\begin{theorem}\label{thm:syndrome_condition_main_statement}
Let $\Fspace$ be defined as in Definition~\ref{def:LRPCcodes}, $t$ be a positive integer with $t \lambda < \min\{m,n-k+1\}$, and $\Espace$ be an error space of rank $t$.
Suppose that the product condition is fulfilled for $\Espace$ and $\Fspace$.
Suppose further that $\H$ has the {maximal-row-span} and {unity properties} (cf.~Definition~\ref{def:H_properties}).

Let $\e$ be an error word, chosen uniformly at random among all error words with support $\Espace$.
Then, the probability that the {syndrome condition} is not fulfilled for $\e$ is
\begin{align*}
\Pr\left( \Sspace \neq \Espace \cdot \Fspace  \mid \phi^{\Espace \cdot \Fspace} = \phi^{\Espace} \phi^{\Fspace} \right) \leq 1
- \prod_{i=0}^{\lambda t -1} \left(1-p^{[i-(n-k)]s}\right) < 4 p^{-s(n-k+1-\lambda t)}.
\end{align*}
\end{theorem}

\begin{proof}
Let $\e' \in \Rqm^n$ be chosen such that every entry $e_i'$ is chosen uniformly at random from the error support $\Espace$.\footnote{This means that $\e'$ might have a support that is contained in, but not equal to $\Espace$. The difference to the actual error $\e$ is that $\e$ is chosen uniformly from all errors of support exactly $\Espace$.}
Denote by $\Sspace_{\e}$ and $\Sspace_{\e'}$ the syndrome spaces obtained by computing the syndromes of $\e$ and $\e'$, respectively.
Then, we have
\begin{align*}
\Pr\big(\Sspace_{\e'} = \Espace \cdot \Fspace\big) \leq \Pr\big( \Sspace_{\e'} = \Espace \cdot \Fspace \mid \supp(\e') = \Espace \big) = \Pr\big( \Sspace_{\e} = \Espace \cdot \Fspace \big),
\end{align*}
where the latter equality follows from the fact that the random experiments of choosing $\e'$ and conditioning on the property that $\e'$ has support $\Espace$ is the same as directly drawing $\e$ uniformly at random from the set of errors with support $\Espace$.
Hence, we obtain a lower bound on $\Pr\big( \Sspace_{\e} = \Espace \cdot \Fspace \big)$ by studying $\Pr\big(\Sspace_{\e'} = \Espace \cdot  \Fspace\big)$, which we do in the following.

Let $f_1,\dots,f_\lambda$ and $\varepsilon_1,\dots,\varepsilon_t$ be \mingensets of $\Fspace$ and $\Espace$, respectively, such that $f_j \varepsilon_i$ for $i=1,\dots,t$, $j=1,\dots,\lambda$ form an \mingenset of $\Espace \cdot \Fspace$. Note that the existence of such bases is guaranteed by the assumed product condition $\phi^{\Espace \cdot \Fspace} = \phi^{\Espace} \phi^{\Fspace}$.

Since $e'_i$ is an element drawn uniformly at random from $\Espace$, we can write it as
$e'_i = \sum_{\mu=1}^{t} e_{i,\mu}' \varepsilon_\mu$,
where $e_{i,j}'$ are uniformly distributed on $\Rq$.
We can assume uniformity of $e_{i,\mu}'$ since for a given $e'_i$, the decomposition of $e_{i,\mu}'$ is unique modulo $\maxIdeal^{r-v(\varepsilon_i)}$. In particular, there are equally many decompositions $[e_{i,1}',\dots,e_{i,t}']$ for each $e'_i$ and the sets of these decompositions are disjoint for different $i$.

Due to the unity property of the parity-check matrix $\H$, we can write any entry $H_{i,j}$ of $\H$ as
$H_{i,j} = \sum_{\eta=1}^{\lambda} h_{i,j,\eta} f_\eta$,
where the $h_{i,j,\eta}$ are units in $\Rq$ or zero.
Furthermore, since each row of $\H$ spans the entire module $\Fspace$ ({full-row-span property}), for each $i$ and each $\eta$, there is at least one $j^*$ with $h_{i,j^*,\eta} \neq 0$. By the previous assumption, this means that $h_{i,j^*,\eta} \in \Rq^*$.

Then, each syndrome coefficient can be written as
\begin{equation*}
s_i = \sum_{j=1}^{n} e'_j H_{i,j}
= \sum_{\mu=1}^{t} \sum_{\eta=1}^{\lambda} \underbrace{\left(\sum_{j=1}^{n}  e_{j,\mu}'  h_{i,j,\eta}\right)}_{=: s_{\mu,\eta,i}} \varepsilon_\mu f_\eta.
\end{equation*}

By the above discussion, for each $i$ and $\eta$, there is a $j^*$ with $h_{i,j^*,\eta} \in \Rq^*$. Hence, $s_{\mu,\eta,i}$ is a sum (with at least one summand) of the products of uniformly distributed elements of $\Rq$ and units of $\Rq$. A uniformly distributed ring element times a unit is also uniformly distributed on $\Rq$. Hence $s_{\mu,\eta,i}$ is a sum (with at least one summand) of uniformly distributed elements of $\Rq$. Hence, $s_{\mu,\eta,i}$ itself is uniformly distributed on $\Rq$.

All together, we can write
\begin{align*}
\begin{bmatrix}
s_1 \\
s_2 \\
\vdots \\
s_{n-k}
\end{bmatrix}
= 
\underbrace{
\begin{bmatrix}
s_{1,1,1} & s_{1,2,1} & \dots & s_{t,\lambda,1} \\
s_{1,1,2} & s_{1,2,2} & \dots & s_{t,\lambda,2} \\
\vdots & \vdots & \ddots & \vdots \\
s_{1,1,n-k} & s_{1,2,n-k} & \dots & s_{t,\lambda,n-k} \\
\end{bmatrix}}_{=: \, \S}
\cdot 
\begin{bmatrix}
\varepsilon_1 f_1 \\
\varepsilon_1 f_2 \\
\vdots \\
\varepsilon_t f_\lambda \\
\end{bmatrix},
\end{align*}
where, by assumption, the $\varepsilon_i f_j$ are a generating set of $\Espace \cdot \Fspace$ and the matrix $\S$ is chosen uniformly at random from $\Rq^{(n-k)\times t \lambda}$.
If $\S$ has full free rank $t \lambda$, then we have $\Sspace_{\e'} = \Espace \cdot \Fspace$.
By Lemma~\ref{lem:syndrome_number_of_full-rank_matrices}, the probability of drawing such a full-rank matrix is
\begin{align*}
\frac{\NM(a,b;\Rq)}{|\Rq|^{ab}} = \prod_{a'=0}^{a-1} \left(1-p^{(a'-b)s} \right).
\end{align*}
This proves the bound
\begin{align*}
\Pr\left( \Sspace \neq \Espace \cdot \Fspace  \mid \phi^{\Espace \cdot \Fspace} = \phi^{\Espace} \phi^{\Fspace} \right) \leq 1
- \prod_{i=0}^{\lambda t -1} \left(1-p^{[i-(n-k)]s}\right).
\end{align*}
We simplify the bound further using the observation that the product is a $q$-Pochhammer symbol. Hence, we have
\begin{align*}
1 - \prod_{i=0}^{\lambda t -1} \left(1-p^{[i-(n-k)]s}\right) = \sum_{j=1}^{\lambda t} \underbrace{(-1)^{j+1} p^{-j(n-k)s} \qbin{\lambda t}{j}{p^s} p^{s\binom{j}{2}}}_{=: \, a_j},
\end{align*}
where $\qbin{a}{b}{q} := \prod_{j=1}^{b} \tfrac{q^{a+1-j}-1}{q^{j}-1}$ is the Gaussian binomial coefficient. Using $q^{b(a-b)} \leq \qbin{a}{b}{q} < 4 q^{b(a-b)}$, we obtain
\begin{align*}
\left| \frac{a_{j+1}}{a_j}\right| &= p^{-(n-k-j)s} \frac{\qbin{\lambda t}{j+1}{p^s}}{\qbin{\lambda t}{j}{p^s}}
< p^{-(n-k-j)s} \frac{4 q^{s(j+1)(\lambda t-j-1)}}{q^{sj(\lambda t-j)}} \\
&= 4 p^{s[\lambda t - j - (n-k+1)]}
< 1
\end{align*}
for $\lambda t < n-k+1$, i.e., $|a_j|$ is strictly monotonically decreasing. Since the summands $a_j$ have alternating sign, we can thus bound $\sum_{j=1}^{\lambda t}a_j \leq a_1$, which gives
\begin{align*}
1 - \prod_{i=0}^{\lambda t -1} \left(1-p^{[i-(n-k)]s}\right) \leq a_1 < 4 p^{-s(n-k+1-\lambda t)}
\end{align*}
\qed
\end{proof}

\begin{remark}\label{rem:necessity_of_maximal_row_span_condition_or_similar}
In contrast to Theorem~\ref{thm:syndrome_condition_main_statement} the {full-row-span property} was not assumed in \cite[Proposition~4.3]{aragon2019low}, which is the analogous statement for finite fields.
However, also the statement in \cite[Proposition~4.3]{aragon2019low} is only correct if we assume additional structure on the parity-check matrix
(\emph{e.g.}, that each row spans the entire space $\Fspace$ or a weaker condition), due to the following counterexample:
Consider a parity-check matrix $\H$ that contains only non-zero entries on its diagonal and in the last row, where the diagonal entries are all $f_1$ and the last row contains the remaining $f_2,\dots,f_\lambda$, i.e.,
\begin{align*}
\setcounter{MaxMatrixCols}{20}
\H := 
\begin{bmatrix}
f_1 & 0 & \dots & 0 & 0 & 0 & \dots & 0 & 0 & \dots & 0 \\
0 & f_1 & \dots & 0 & 0 & 0 & \dots & 0 & 0 & \dots & 0 \\
\vdots & \vdots & \ddots & \vdots & \vdots & \vdots & \ddots & \vdots & \vdots & \ddots & \vdots \\
0 & 0 & \dots & f_1 & f_2 & f_3 & \dots & f_\lambda & 0 & \dots & 0
\end{bmatrix}.
\end{align*}
This is a valid parity-check matrix according to \cite[Definition~4.1]{aragon2019low} since the entries of $\H$ span the entire space $\Fspace$.
However, due to the structure of the matrix, the first $n-k-1$ syndromes are all in $f_1 \Espace$, hence $\rk_{\Rq}(\Sspace) \leq t+1 < t \lambda$ for \emph{any} error of support~$\Espace$.
\end{remark}

\subsection{Failure of Intersection Condition}\label{ssec:failure_intersection}

We use a similar proof strategy as in \cite{aragon2019low} to derive an upper bound on the failure probability of the intersection condition.
The following lemma is the Galois-ring analog of \cite[Lemma~3.4]{aragon2019low}, where the difference is that we need to take care of the fact that the representation of module elements in an \mingenset is not necessarily unique in a Galois ring.

\begin{lemma}\label{lem:intersection_failure_lemma}
Let $\Aspace \subseteq \Rqm$ be an $\Rq$-module of rank $\alpha$ and $\Bspace \subseteq \Rqm$ be a free $\Rq$-module of free rank $\beta$.
Assume that
$\phi^{\Aspace \cdot \Bspace^2} = \phi^{\Aspace} \phi^{\Bspace^2}$ and that there is an element $e \in \Aspace \cdot \Bspace \setminus \Aspace$ with $e \Bspace \subseteq \Aspace \cdot \Bspace$.
Then, there is an $y \in \Bspace \setminus \Rq$ such that $y \Bspace \subseteq \Bspace$.
\end{lemma}

\begin{proof}
Let $a_1,\dots,a_\alpha$ be an \mingenset of $\Aspace$ and $b_1,\dots,b_\beta$ be a basis of $\Bspace$.
Due to $e \in \Aspace \cdot  \Bspace$, there are coefficients $e_{i,j} \in \Rq$ such that
\begin{align}
\textstyle e = \sum_{i=1}^{\alpha} \underbrace{\left(\textstyle\sum_{j=1}^{\beta} e_{i,j} b_j \right)}_{=: \, b'_i} a_i. \label{eq:e_unique_representation}
\end{align}
Due to the fact that $e \notin \Aspace$, there is an $\eta \in \{1,\dots,\alpha\}$ with $b_\eta' a_\eta \notin \Aspace$. In particular, $y := g_\maxIdeal^{v(a_\eta)} b_\eta' \in \Bspace \setminus \Rq$. We show that $y$ fulfills $y \Bspace \subseteq \Bspace$.

Let now $b \in \Bspace$. Since by assumption $eb \in \Aspace \cdot \Bspace$, there are $c_{i,j} \in \Rq$ with
$e b = \sum_{i=1}^{\alpha} \left(\sum_{j=1}^{\beta} c_{i,j} b_j \right) a_i$.
By \eqref{eq:e_unique_representation}, we can also write 
$e b = \sum_{i=1}^{\alpha} \left(\sum_{j=1}^{\beta} e_{i,j} b_j b \right) a_i = \sum_{i=1}^{\alpha} b_i' b a_i$.
Due to the maximality of the rank profile of $\Aspace \cdot \Bspace^2$, i.e., $\phi^{\Aspace \cdot \Bspace^2} = \phi^{\Aspace} \phi^{\Bspace^2}$, we have that the coefficients $c_{i} \in \Bspace^2$ of any representation $c = \sum_i c_i a_i$ of an element $c \in \Aspace \cdot \Bspace^2$ are unique modulo $\MaxIdeal^{r-v(a_i)}$.
Hence, for every $i=1,\dots,\alpha$, there exists $\chi_i \in \Bspace^2$ such that
\begin{align*}
b_i' b = \sum_{j=1}^{\beta} c_{i,j} b_j + g_\maxIdeal^{r-v(a_i)} \chi_i.
\end{align*}
Thus, with $\sum_{j=1}^{\beta} c_{\eta,j} b_j \in \Bspace$, $g_\maxIdeal^{v(a_i)} \in \Rq$, and $g_\maxIdeal^{r}=0$, we get
\begin{align*}
y b = g_\maxIdeal^{v(a_\eta)} b_\eta' b =
g_\maxIdeal^{v(a_\eta)}\sum_{j=1}^{\beta} c_{\eta,j} b_j + g_\maxIdeal^{r} \chi_\eta \in \Bspace.
\end{align*}
Since this hold for any $b$, we have $y \Bspace \subseteq \Bspace$, which proves the claim.
\qed
\end{proof}

We get the following bound using Lemma~\ref{lem:intersection_failure_lemma}, Theorem~\ref{thm:general_product}, and a similar argument as in \cite{gaborit2013low}.

\begin{theorem}\label{thm:intersection_failure_main_statement}
Let $\Fspace$ be defined as in Definition~\ref{def:LRPCcodes} such that it has the base-ring property (i.e., $1 \in \Fspace$).
Suppose that no intermediate ring $R'$ between $\Rq \subsetneq R' \subseteq \Rqm$ is contained in $\Fspace$ (this holds, e.g., for $\lambda$ greater than the smallest divisor of $m$ or for special $\Fspace$).

Let $t$ be a positive integer with $t \tfrac{\lambda(\lambda+1)}{2} < m$ and $t \lambda < n-k+1$, and let $\phi(x) \in \ZZ[x]/(x^r)$ with nonnegative coefficients such that $\phi(1)=t$.
Choose $\e \in \Rqm^n$ uniformly at random from the set of vectors with whose support has rank profile $\phi$.

Then, the probability that the {intersection condition} is not fulfilled, given that syndrome and product conditions are satisfied, is
\begin{align*}
&\Pr\left(\textstyle \bigcap_{i=1}^{\lambda} \Sspace_i = \Espace \mid \Sspace = \Espace \cdot  \Fspace  \land \phi^{\Espace \cdot \Fspace} = \phi^{\Espace} \phi^{\Fspace}  \right) \\
&\leq \left(1-p^{-s\frac{\lambda(\lambda+1)}{2}}\right) \sum_{i=1}^{t} \sum_{j = 0}^{r-1} p^{s(r-j)\left(i \frac{\lambda(\lambda+1)}{2}-m\right)} \leq 2 t p^{s\left(t \frac{\lambda(\lambda+1)}{2}-m\right)}
\end{align*}
\end{theorem}

\begin{proof}
Suppose that the product ($\phi^{\Espace \cdot \Fspace} = \phi^{\Espace} \phi^{\Fspace}$) and syndrome ($\Sspace = \Espace \cdot \Fspace$) conditions are fulfilled, and assume that the intersection condition is not fulfilled.
Then we have
$\bigcap_{i=1}^{\lambda} \Sspace_i =: \Espace' \supsetneq \Espace$.
Choose any $e \in \Espace' \setminus \Espace$.
Since $\Fspace$ contains $1$ by assumption, we have $e \in \Aspace \cdot \Bspace$.
Due to $\Aspace \subseteq \Espace$, we have $e \notin \Aspace$.
Furthermore, we have $\Espace' \cdot \Bspace = \Espace \cdot \Bspace$, so all conditions on $e$ of Lemma~\ref{lem:intersection_failure_lemma} are fulfilled.

Since $\Espace$ is chosen uniformly at random from all free submodules of $\Rqm$ of rank $t$, we can apply Theorem~\ref{thm:general_product} and obtain that
$\phi^{\Espace \cdot \Fspace^2} = \phi^{\Espace} \phi^{\Fspace^2}$
with probability at least
\begin{align*}
&\Pr\!\left(\phi^{\Aspace \cdot \Bspace^2} \neq \phi^{\Aspace} \phi^{\Bspace^2}\right) \\
&\leq \left(1-p^{-s\lambda'}\right) \sum_{i=1}^{t} \sum_{j = 0}^{r-1} p^{s(r-j)(i \lambda'-m)}\\
&\leq \left(1-p^{-s\frac{\lambda(\lambda+1)}{2}}\right) \sum_{i=1}^{t} \sum_{j = 0}^{r-1} p^{s(r-j)\left(i \frac{\lambda(\lambda+1)}{2}-m\right)} \\
&\leq 2 t p^{s\left(t \frac{\lambda(\lambda+1)}{2}-m\right)}
\end{align*}
where $\lambda' := \rk_{\Rq}(\Fspace^2) \leq \tfrac{1}{2}\lambda(\lambda+1)$ (this is clear since $\Fspace^2$ is generated by the products of all unordered element pairs of an \mingenset of $\Fspace$).

Hence, with probability at least one minus this value, both conditions of Lemma~\ref{lem:intersection_failure_lemma} are fulfilled.
In that case, there is an element $y \in \Fspace \setminus \Rq$ such that $y \Fspace \subseteq \Fspace$.
Thus, also $y^i \Fspace \subseteq \Fspace$ for all positive integers $i$, and we have that the ring $\Rq(y)$ extended by the element $y \notin \Rq$ fulfills $\Rq(y) \subseteq \Fspace$ (this holds since $\Fspace$ contains at least one unit).
This is a contradiction to the assumption on intermediate rings.
\qed
\end{proof}

\subsection{Overall Failure Probability}

The following theorem states the overall bound on the failure probability, exploiting the  bounds derived in Theorems~\ref{thm:product_condition_main_statement}, \ref{thm:syndrome_condition_main_statement}, and~\ref{thm:intersection_failure_main_statement}.

\begin{theorem}\label{thm:main_statement}
Let $\Fspace$ be defined as in Defintion~\ref{def:LRPCcodes} such that it has the base-ring property (i.e., $1 \in \Fspace$).
Suppose that no intermediate ring $R'$ between $\Rq \subsetneq R' \subseteq \Rqm$ is contained in $\Fspace$ (this holds, e.g., for $\lambda$ greater than the smallest divisor of $m$ or for special $\Fspace$).
Suppose further that $\H$ has the {maximal-row-span} and {unity properties} (cf.~Definition~\ref{def:H_properties}).

Let $t$ be a positive integer with $t \tfrac{\lambda(\lambda+1)}{2} < m$ and $t \lambda < n-k+1$, and let $\phi(x) \in \ZZ[x]/(x^r)$ with nonnegative coefficients such that $\phi(1)=t$.
Choose $\e \in \Rqm^n$ uniformly at random from the set of vectors with whose support has rank profile $\phi$.

Then, Algorithm~\ref{alg:decoder} with input $\c+\e$ returns $\c$ with a failure probability of at most
\begin{align}
\Pr(\text{failure}) &\leq \left(1-p^{-s\lambda}\right) \sum_{i=1}^{t} \sum_{j = 0}^{r-1} p^{s(r-j)\left(i \lambda-m\right)} \notag \\
&\quad + \left[1- \prod_{i=0}^{\lambda t -1} \left(1-p^{[i-(n-k)]s}\right)\right] \notag \\
&\quad + \left(1-p^{-s\frac{\lambda(\lambda+1)}{2}}\right) \sum_{i=1}^{t} \sum_{j = 0}^{r-1} p^{s(r-j)\left(i \frac{\lambda(\lambda+1)}{2}-m\right)} \label{eq:bound_complicated_but_tighter} \\
&\leq 4 p^{s[\lambda t-(n-k+1)]} + 4 t p^{s\left(t \frac{\lambda(\lambda+1)}{2}-m\right)} \label{eq:bound_simple_but_loose}
\end{align}
\end{theorem}

\begin{proof}
The statement follows by applying the union bound to the failure probabilities of the three success conditions, derived in Theorems~\ref{thm:product_condition_main_statement}, \ref{thm:syndrome_condition_main_statement}, and~\ref{thm:intersection_failure_main_statement}.
\qed
\end{proof}

The simplified bound \eqref{eq:bound_simple_but_loose} in Theorem~\ref{thm:main_statement} coincides up to a constant with the bound by Gaborit et at.\ \cite{gaborit2013low} in the case of a finite field (Galois ring with $r=1$).
If we compare an LRPC code over a finite field of size $p^{rs}$ and with an LRPC code over a Galois ring with parameters $p,r,s$ (i.e., the same cardinality), then we can observe that the bounds have the same exponent, but the base of the exponent is different: It is $p^{rs}$ for the field and $p^s$ for the ring case. Hence, the maximal decoding radii $t_\mathrm{max}$ (i.e., the maximal rank $t$ for which the bound is $<1$) are roughly the same, but the exponential decay in $t_\mathrm{max}-t$ for smaller error rank $t$ is slower in case of rings due to a smaller base of the exponential expression.
This ``loss'' is expected due to the weaker structure of modules over Galois rings compared to vector spaces over fields.

\section{Decoding Complexity}\label{sec:complexity}

We discuss the decoding complexity of the decoding algorithm described in Section~\ref{sec:decoding}.
Over a field, all operations within the decoding algorithm are well-studied and it is clear that the algorithm runs in roughly $\softO(\lambda^2 n^2 m)$ operations over the small field $\Fq$.
Although we believe that an analog treatment over the rings studied in this paper must be known in the community, we have not found a comprehensive complexity overview of the corresponding operations in the literature.
Hence, we start the complexity analysis with an overview of complexities of ring operations and linear algebra over these rings.

\subsection{Cost Model and Basic Ring Operations}

We express complexities in operations in $\Rq$.
For some complexity expressions, we use the soft-O notation, i.e., $f(n) \in \softO(g(n))$ if there is a $r \in \ZZ_{\geq 0}$ such that $f(n) \in \softO(g(n) \log(g(n))^r)$.
We use the following result, which follows straightforwardly from standard computer-algebra methods in the literature.

\begin{lemma}[Collection of results in \cite{von2013modern}]
Addition in $\Rqm$ costs $m$ additions in $\Rq$.
Multiplication in $\Rqm$ can be done in $O(m \log(m) \log(\log(m)))$ operations in $\Rq$.
\end{lemma}

\begin{proof}
We represent elements of $\Rqm$ as residue classes of polynomials in $\Rq[z]/(h(z))$ (e.g., each residue class is represented by its unique representative of degree $<m$), where $h \in \Rq[z]$ is a monic polynomial of degree $m$ as explained in the preliminaries.

Addition is done independently on the $m$ coefficients of the polynomial representation, so it only requires $m$ additions in $\Rq$.
Multiplication consists of multiplying two residue classes in $\Rq[z]/(h(z))$, which can be done by multiplying the two representatives of degree $<m$ and then taking them modulo $(h(z))$ (i.e., take the remainder of the division by the monic polynomial $h$). Both multiplication and division can be implemented in $O(m \log(m) \log(\log(m)))$ time using Sch\"onhage and Strassen's polynomial multiplication algorithm (cf.~\cite[Section~8.3]{von2013modern}) and a reduction of division to multiplication using a Newton iteration (cf.~\cite[Section~9.1]{von2013modern}).
Note that both methods work over any commutative ring with $1$.
\qed
\end{proof}

\subsection{Linear Algebra over Galois Rings}

We recall how fast we can compute the Smith normal form of a matrix over $\Rq$ and show that computing the right kernel of a matrix and solving a linear system can be done in a similar speed.
Let $2 \leq \omega\leq 3$ be the matrix multiplication exponent (e.g., $\omega = 2.37$ using the Coppersmith--Winograd algorithm).

\begin{lemma}[\!\!{\cite[Proposition~7.16]{storjohann2000algorithms}}]\label{lem:complexity_smith_normal_form}
Let $\A \in \Rq^{a \times b}$. Then, the Smith normal form $\D$ of $\A$, as well as the corresponding transformation matrices $\S$ and $\T$, can be computed in
\begin{equation*}
O(a b \min\{a,b\}^{\omega-2} \log(a+b))
\end{equation*}
operations in $\Rq$.
\end{lemma}

\begin{lemma}\label{lem:complexity_right_kernel}
Let $\A \in \Rq^{a \times b}$. An \mingenset of the right kernel of $\A$ can be computed in $O(a b \min\{a,b\}^{\omega-2} \log(a+b))$ operations in $\Rq$.
\end{lemma}

\begin{proof}
We compute the Smith normal form $\D = \S \A \T$ and the transformation matrices $\S$ and $\T$ of $\A$.
To compute the right kernel, we need to solve the homogeneous linear system $\A \x = \0$ for $\x$.
Using the Smith normal form, we can rewrite it into
\begin{align*}
\D \T^{-1} \x = \0.
\end{align*}
Denote $\y := \T^{-1} \x$ and first solve $\D \y = \0$.
W.l.o.g., let the diagonal entries of $\D$ be of the form
\begin{equation*}
\begin{bmatrix} \I_{n_0} & & & & \\
  & g_\maxIdeal \I_{n_1} & & & \\
  & & \ddots & & \\
  & & & g_{\maxIdeal}^{r-1}\I_{n_{r-1}} & \\
  & & &  & \0
 \end{bmatrix}
\end{equation*}
where the $n_i$ are the coefficients of the rank profile $\phi(x)=\sum_{i=0}^{r-1}n_ix^i\in \mathbb{N}[x]/(x^r)$ of $\A$'s row space.
Then, the rows of the following matrix are an \mingenset of the right kernel of $\D$ (we denote by $\eta := n_0$ the free rank of $\A$'s row space and by $\mu := \sum_{i=0}^{r-1}n_i)$ the rank of $\A$'s row space):
\begin{align*}
\K := \begin{bmatrix}
\0_{(\mu-\eta) \times \eta} & \B & \0_{(\mu-\eta) \times (b-\mu)} \\
\0_{(b-\mu) \times \eta} & \0_{(b-\mu) \times (\mu-\eta)} & \I_{(b-\mu) \times (b-\mu)} \\
\end{bmatrix} \in \Rq^{(b-\eta) \times b},
\end{align*}
where
\begin{align*}
\B :=
\begin{bmatrix}
g_\maxIdeal^{r-1} \I_{n_1} & & & \\
& g_\maxIdeal^{r-2} \I_{n_1} & & \\
&&  \ddots & \\
& & & g_{\maxIdeal}^{1}\I_{n_{r-1}} \\
\end{bmatrix}.
\end{align*}

Hence, the rows of $\K \T^\top$ form an \mingenset of the right kernel of $\A$. Note that this matrix multiplication can be implemented with complexity $O(b^2)$ since $\K$ has only at most one entry per row and column.
\qed
\end{proof}

\begin{lemma}\label{lem:complexity_find_one_solution_of_one_system}
Let $\A \in \Rq^{a \times b}$ and $\b \in \Rq^{a}$. A solution of the linear system $\A \x = \b$ (or, in case no solution exists, the information that it does not exist) can be obtained in $O(a b \min\{a,b\}^{\omega-2} \log(a+b))$ operations in $\Rq$.
\end{lemma}

\begin{proof}
We follow the same strategy and the notation as in Lemma~\ref{lem:complexity_right_kernel}.
Solve
\begin{align*}
\D \underbrace{\T^{-1} \x}_{=: \, \y} = \S \b =: \b'.
\end{align*}
for one $\y$. The system has a solution if and only if $b_j' \in \MaxIdeal^{i_j}$ for $j=1,\dots,r'$, and $b_j' = 0$ for all $j>r'$. In case it has a solution, it is easy to obtain a solution $\y$. Then we only need to compute $\x = \T \y$, which is a solution of $\A \x = \b$.
The heaviest step is to compute the Smith normal form, which proves the complexity statement.
\qed
\end{proof}

\subsection{Complexity of the LRPC Decoder over Galois Rings}

\begin{theorem}\label{thm:complexity}
Suppose that the inverse elements $f_1^{-1},\dots,f_\lambda^{-1}$ are precomputed.
Then, Algorithm~\ref{alg:decoder} has complexity $\softO(\lambda^2 n^2 m)$ operations in $\Rq$.
\end{theorem}

\begin{proof}
The heaviest steps of Algorithm~\ref{alg:decoder} (see Section~\ref{sec:decoding}) are as follows:

Line~\ref{line:syndrome_computation} computes the syndrome $\s$ from the received word. This is a vector-matrix multiplication in $\Rqm$, which costs $O(n(n-k)) \subseteq O(n^2)$ operations in $\Rqm$, i.e., $\softO(n^2m)$ operations in $\Rq$.

Line~\ref{line:S_i_computation} is called $\lambda$ times and computes for each $f_i$ the set $S_i = f_i^{-1} \Sspace$ (recall that the inverses $f_i^{-1}$ are precomputed).
We obtain a generating set of $\Sspace_i$ by multiplying $f_i^{-1}$ to all syndrome coefficients $s_1,\dots,s_{n-k}$. This costs $O(\lambda(n-k))$ operations in $\Rqm$ in total, i.e., $\softO(\lambda n m)$ operations in $\Rq$. If we want a minimal generating set, we can compute the Smith normal form for each $\Sspace_i$, which costs $\softO(\lambda n^{\omega-1}m)$ operations in $\Rq$ according to Lemma~\ref{lem:complexity_smith_normal_form}.

Line~\ref{line:S_i_intersection} computes the intersection $\Espace' \gets \bigcap_{i=1}^{\lambda} \Sspace_i$ of the modules $\Sspace_i$.
This can be computed via the kernel computation algorithm as follows: Let $\Aspace$ and $\Bspace$ be two modules. Then, we have $\Aspace \cap \Bspace = \ker\left( \ker(\Aspace) \cup \ker(\Bspace) \right)$. Hence, we can compute the intersection $\Aspace \cap \Bspace$ by writing generating sets of the modules as the rows of two matrices $\A$ and $\B$, respectively. Then, we compute matrices $\A'$ and $\B'$, whose rows are generating sets of the right kernel of $\A$ and $\B$, respectively. Then, rows of the matrix $\C := \begin{bmatrix}
\A' \\
\B'
\end{bmatrix}$
are a generating set of $\ker(\Aspace) \cup \ker(\Bspace)$, and be obtain $\Aspace \cap \Bspace$ by computing again the right kernel of $\C$. By applying this algorithm iteratively to the $\Sspace_i$ (using the kernel computation algorithm described in Lemma~\ref{lem:complexity_right_kernel}), we obtain the intersection $\Espace'$ in $\softO(\lambda n^{\omega-1}m)$ operations.

Line~\ref{line:erasure_decoding} recovers an error vector $\e$ from the support $\Espace'$ and syndrome $\s$. As shown in the proof of Lemma~\ref{lem:erasure_decoding}, this can be done by solving $t$ linear systems over $\Rq$ with each $n$ unknowns and $(n-k)\lambda$ equations w.r.t.\ the same matrix $\H_{\mathrm{ext}}$. Hence, we only once need to compute the Smith normal form of $\H_{\mathrm{ext}}$, which requires $\softO(n [(n-k)\lambda]^{\omega-1})$ operations. The remaining steps for solving the systems (see~Lemma~\ref{lem:complexity_find_one_solution_of_one_system} to compute one solution, if it exists, and Lemma~\ref{lem:complexity_right_kernel} to compute an affine basis) consist mainly of matrix-vector operations, which require in total $\softO(t \lambda^2(n-k)^2)$ operations in $\Rq$, where $t \leq m$ is the rank of $\Espace'$.
Note that during the algorithm, it is easy to detect whether the systems have no solution, a unique solution, or more than one solution.
\qed
\end{proof}

\begin{remark}
The assumption that $f_1^{-1},\dots,f_\lambda^{-1}$ are precomputed makes sense since in many application, the code is chosen once and then several received words are decoded for the same $f_1,\dots,f_\lambda$.
Precomputation of all $f_1^{-1},\dots,f_\lambda^{-1}$ costs at most $\softO(\lambda m^\omega)$ since for $a \in \Rqm$, the relation $a^{-1} a \equiv 1 \mod h$ (for $a$ and $a^{-1}$ being the unique representative in $\Rq[z]/(h)$ with degree $<m$) gives a linear system of equations of size $m \times m$ over $\Rq$ with a unique solution $a^{-1}$. This complexity can only exceed the cost bound in Theorem~\ref{thm:complexity} if $m \gg n$.

In fact, we conjecture, but cannot rigorously prove, that the inverse of a unit in $\Rqm$ can be computed in $\softO(m)$ operations in $\Rq$ using a fast implementation of the extended Euclidean algorithm (see, e.g., \cite{von2013modern}).
If this is true, the precomputation cost is smaller than the cost bound in Theorem~\ref{thm:complexity}.
\end{remark}

The currently fastest decoder for Gabidulin codes over finite rings, the Welch--Berlekamp-like decoder in \cite{kamche2019rank}, has complexity $O(n^\omega)$ operations over $\Rqm$ since its main step is to solve a linear system of equations. Over $\Rq$, this complexity bound is $\softO(n^\omega m)$, i.e., it is larger than the complexity bound for our LRPC decoder for constant $\lambda$ and the same parameters $n$ and $m$.

\section{Simulation Results}
\label{sec:simulations}

We performed simulations of LRPC codes with $\lambda=2$, $k=8$ and $n=20$ (note that we need $k \leq \tfrac{\lambda-1}{\lambda}n$ by the unique-decoding property) over the ring $\Rqm$ with $p=r=2$, $s=1$ and $m=21$. In each simulation, we generated one parity-check matrix (fulfilling the maximal-row-span and the unity properties) and conducted a Monte Carlo simulation in which we collected at least $1000$ decoding errors and at least $50$ failures of every success condition. 
All simulations gave very similar results and confirmed our analysis. We present one of the simulation results in Figure~\ref{fig:sim} for errors of rank weight $t=1,\hdots,7$ and three different rank profiles.

We indicate by markers the estimated probabilities of violating the product condition (S: Prod), the syndrome condition (S: Synd), the intersection condition (S: Inter) as well as the decoding failure rate (S: Dec). Black markers denote the result of the simulations with errors of rank profile $\phi_1(x) = t$, blue markers show the result with errors of rank profile $\phi_2(x) = tx$ and orange markers indicate the result with rank profile $\phi_3(x) \in \{1,1+x,2+x,2+2x,3+2x,3+3x,4+3x \}$. 
Further, we show the derived bounds\footnote{In Figure~\ref{fig:sim}, we show for each condition the tightest bound that we derived.} on the probabilities of not fulfilling the product condition (B: Prod) given in Theorem~\ref{thm:product_condition_main_statement}, the syndrome condition (B: Synd) derived in Theorem~\ref{thm:syndrome_condition_main_statement}, the intersection condition (B: Inter) provided in Theorem~\ref{thm:intersection_failure_main_statement} and the union bound (B: Dec) stated in Theorem~\ref{thm:main_statement}. Since the derived bounds depend only on the rank weight $t$ but not on the rank profile, we show each bound only once. 

One can observe that the bound on the probability of not fulfilling the syndrome condition is very close to the true probability while the bounds on the probabilities of violating the product and syndrome condition are loose. Gaborit \emph{et al.} have made the same observation in the case of finite fields. In addition, it seems that only the rank weight but not the rank profile has an impact on the probabilities of violating the success conditions.

\begin{figure}[h]
\includegraphics{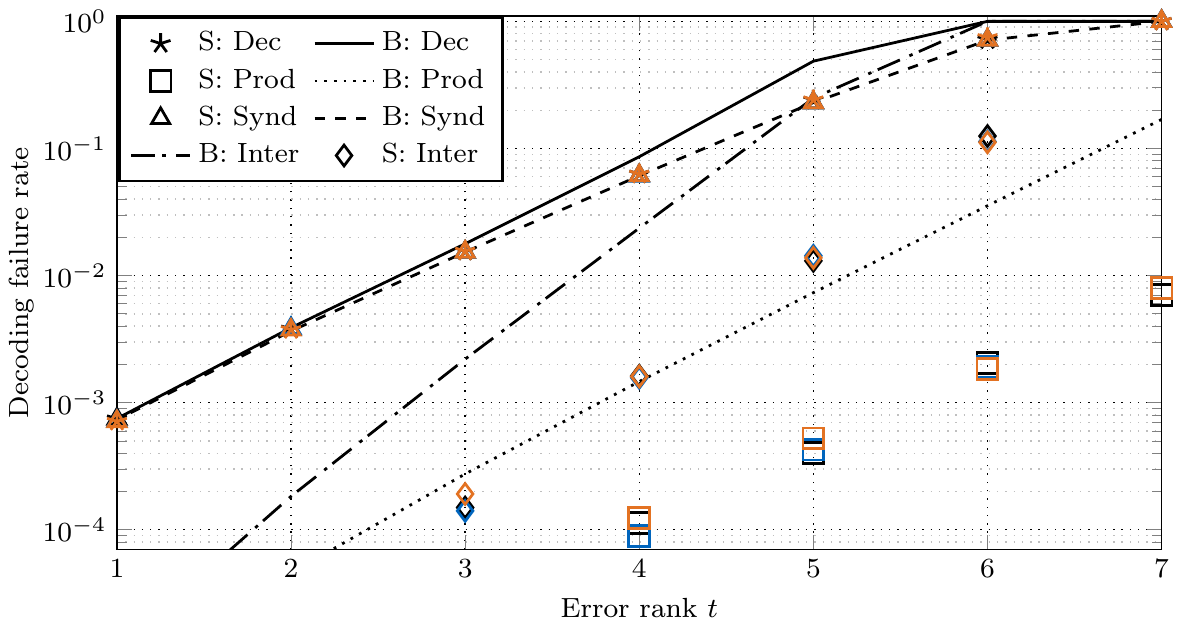}
  \caption{Simulation results for $\lambda=2$, $k=8$ and $n=20$ over $\Rqm$ with $p=r=2$, $s=1$ and $m=21$. The markers indicate the estimated probabilities of not fulfilling the product condition (S: Prod), the syndrome condition (S: Synd), the intersection condition (S: Inter) and the decoding failure rate (S: Dec), where the black, blue and orange markers refer to errors of rank profile $\phi_1(x) =t$, $\phi_2(x) =tx$ and $\phi_3(x)\in\{ 1,1+x,2+x,2+2x,3+2x,3+3x,4+3x \}$, respectively. The derived bounds on these probabilities are shown as lines.}
\label{fig:sim}
\end{figure}

We also found that the base-ring property of $\Fspace$ is---in all tested cases---not necessary for the failure probability bound on the intersection condition (Theorem~\ref{thm:intersection_failure_main_statement}) to hold.
It is an interesting question whether we can prove the bound without this assumption, both for finite fields and rings.

\section{Conclusion}

We have adapted low-rank parity-check codes from finite fields to Galois rings and showed that Gaborit et al.'s decoding algorithm works as well for these codes.
We also presented a failure probability bound for the decoder, whose derivation is significantly more involved than the finite-field analog due to the weaker structure of modules over finite rings.
The bound shows that the codes have the same maximal decoding radius as their finite-field counterparts, but the exponential decay of the failure bound has $p^s$ as a basis instead of the cardinality of the base ring $|\Rq|=p^{rs}$ (note $\Rq$ is a finite field if and only if $r=1$).
This means that there is a ``loss'' in failure probability when going from finite fields to finite rings, which  can be expected due to the zero divisors in the ring.

The results show that LRPC codes work over finite rings, and thus can be considered, as an alternative to Gabidulin codes over finite rings, for potential applications of rank-metric codes, such as network coding and space-time codes---recall from the introduction that network and space-time coding over rings may have advantages compared to the case of fields.
It also opens up the possibility to consider the codes for cryptographic applications, the main motivation for LRPC codes over fields.

Open problems are a generalization of the codes to more general rings (such as principal ideal rings); an analysis of the codes in potential applications; as well as an adaption of the improved decoder for LRPC codes over finite fields in \cite{aragon2019low} to finite rings.
To be useful for network coding (both in case of fields and rings), the decoder must be extended to handle row and column erasures in the rank metric (cf.~\cite{silva2008rank,kamche2019rank}).

\section*{Acknowledgment}

The work of J.~Renner was supported by the European Research Council (ERC) under the European Union’s Horizon 2020 research and innovation programme (grant agreement no.~801434).

A. Neri was supported by the Swiss National Science Foundation through grant no.~187711.

S.~Puchinger received funding from the European Union's Horizon 2020 research and innovation program under the Marie Sklodowska-Curie grant agreement no.~713683.

\bibliographystyle{splncs04}
\bibliography{main}

\appendix
\begin{appendices}
\section{Proofs of Corollaries~\ref{cor:num_free} and \ref{cor:num_mod}}\label{app:proofs_freemodule_corollaries}

In this section we provide the proofs of Corollaries~\ref{cor:num_free} and \ref{cor:num_mod} in Section~\ref{ssec:failure_product}.

Inspired by Proposition \ref{prop:freemodule_matrix}, we study the following notions.
For a given potential rank profile $\phi(x)=\sum_{i=0}^{r-1}n_ix^i\in \mathbb{N}[x]/(x^r)$, with $\phi(1)=N\leq m$, we consider the sets

  \begin{align*} G_\phi&:=\left\{\begin{bmatrix} \Y_{0,0} & g_\maxIdeal \Y_{0,1} & g_\maxIdeal^2 \Y_{0,2} & \cdots & g_\maxIdeal^{r-1} \Y_{0,r-1} \\ 
 \Y_{1,0} & \Y_{1,1} & g_\maxIdeal \Y_{1,2} & \cdots & g_\maxIdeal^{r-2} \Y_{1,r-1} \\
 \Y_{2,0} & \Y_{2,1} &  \Y_{2,2} & \cdots & g_\maxIdeal^{r-3} \Y_{2,r-1} \\
 \vdots & \vdots & \vdots & &\vdots \\
 \Y_{r-1,0} & \Y_{r-1,1} & \Y_{r-1,2} & \cdots & \Y_{r-1,r-1} \\
 \end{bmatrix} : \Y_{i,j}\in\Rq^{n_i\times n_j}  \right\},\\
      G_\phi^*&:=\left\{\begin{bmatrix} \Y_{0,0} & g_\maxIdeal \Y_{0,1} & g_\maxIdeal^2 \Y_{0,2} & \cdots & g_\maxIdeal^{r-1} \Y_{0,r-1} \\ 
 \Y_{1,0} & \Y_{1,1} & g_\maxIdeal \Y_{1,2} & \cdots & g_\maxIdeal^{r-2} \Y_{1,r-1} \\
 \Y_{2,0} & \Y_{2,1} &  \Y_{2,2} & \cdots & g_\maxIdeal^{r-3} \Y_{2,r-1} \\
 \vdots & \vdots & \vdots & &\vdots \\
 \Y_{r-1,0} & \Y_{r-1,1} & \Y_{r-1,2} & \cdots & \Y_{r-1,r-1} \\
 \end{bmatrix} : \Y_{i,j}\in\Rq^{n_i\times n_j}, \Y_{i,i}\in \GL(n_i,\Rq) \right\} \\
  H_\phi&:=\left\{ \begin{bmatrix} 0 \\
 g_\maxIdeal^{r-1}\Z_1 \\ g_\maxIdeal^{r-2}\Z_2 \\ \vdots \\ g_\maxIdeal \Z_{r-1}\end{bmatrix}: \Z_i\in \Rq^{n_i\times m}\right\}.
\end{align*}

Notice that 
\begin{enumerate}
    \item[\mylabel{p1}{(P1)}] $(G_\phi,+,\cdot)$ is a subring of $\Rq^{N\times N}$;
    \item[\mylabel{p2}{(P2)}] $G_\phi^*=G_\phi\cap \GL(N,\Rq)$;
    \item[\mylabel{p3}{(P3)}] $(G_\phi^*,\cdot)$ is a subgroup of $\GL(N,\Rq)$;
    \item[\mylabel{p4}{(P4)}] $(H_\phi,+)$ is a subgroup of $\Rq^{N\times m}$;
    \item[\mylabel{p5}{(P5)}] For every $\Y\in G_\phi$, $\Z\in H_\phi$, we have $\Y\Z\in H_\phi$;
    \item[\mylabel{p6}{(P6)}] If $\Y\in G_\phi^*$, then $\Z\longmapsto \Y\Z$ is a bijection of $H_\phi$.
\end{enumerate} 

With these tools and from Proposition \ref{prop:freemodule_matrix} we can deduce the two corollaries.

\begin{proof}[Proof of Corollary~\ref{cor:num_free}]
  First, denote by $n_i:=\phi_i^{\Mspace}$ and let $N:=n_0+\ldots+n_{r-1}$, and fix an $\Rq$-basis of $\Rqm$ so that we identify $\Rqm$ with $\Rq^m$.
Fix a  free module $\Nspace \in \mathrm{Free}(\Mspace)$ and let $\T_{\Nspace}$ be such that $\rowspace(\T_{\Nspace})=\Nspace$  By Proposition \ref{prop:freemodule_matrix}, we have
  \begin{align*}
      \mathrm{Free}(\Mspace) &=\{\rowspace(\Y\T_{\Nspace}+\Z) \mid \Y\in G_\phi^*,\Z \in H_\phi\}\\
      &=\{\rowspace(\T_{\Nspace}+\Y^{-1}\Z) \mid \Y\in G_{\phi}^*,\Z \in H_\phi\} \\
      &=\{\rowspace(\T_{\Nspace}+\Z) \mid \Z \in H_\phi\},
  \end{align*}
  where the last equality follows from \ref{p6}.
 It is immediate to see that $\rowspace(\T_{\Nspace}+\Z)=\Nspace=\rowspace(\T_{\Nspace})$ if and only if all the rows of $\Z$ belong to $\Nspace$. For the $i$th block of $n_i$ rows of $\Z$, we  can freely choose among all the elements in $g_{\maxIdeal}^{r-i}\Nspace$, that are $s^{iN}$. Hence we get
 \begin{align*}
     |\{ \Z \in H_{\phi} \mid \rowspace(\T_{\Nspace}+\Z)=\Nspace\}|&=|\{\Z\in H_{\phi} \mid \rowspace(\Z) \subseteq \Nspace \}| 
     =\prod_{i=1}^{r-1}s^{in_iN}.
 \end{align*}
 This means that every module is counted $\prod_{i=1}^{r-1}s^{in_iN}$ many times and we finally obtain
 $$ |\mathrm{Free}(\Mspace)|=\frac{|H_\phi|}{\prod_{i=1}^{r-1}s^{in_iN}}=\prod_{i=1}^{r-1}\frac{s^{in_im}}{s^{in_iN}}=s^{(m-N)\sum_{i=1}^{r-1}in_i }. $$
\qed
\end{proof}

\begin{proof}[Proof of Corollary~\ref{cor:num_mod}]
Let $\Mspace$ be an $\Rq$-submodule of $\Rqm$  with rank profile $\phi^{\Mspace}$ and observe that $\Mspace\in \mathrm{Mod}(\phi,\Nspace)$ if and only if $\Nspace \in \mathrm{Free}(\Mspace)$. 
 Identify $\Rqm$ with $\Rq^m$, and define 
\begin{equation*}
 \D:=\begin{bmatrix} \I_{n_0} & & & \\
  & g_\maxIdeal \I_{n_1} & & \\
  & & \ddots & \\
  & & & g_{\maxIdeal}^{r-1}\I_{n_{r-1}}
\end{bmatrix}.
\end{equation*}
With this notation, we have
\begin{equation*}
  \mathrm{Mod}(\phi,\Nspace)=\{\rowspace(\D\T) \mid \T \in \Rq^{N\times m}, \rowspace(\T)=\Nspace\}.
\end{equation*}
  Moreover, there are exactly $|\GL(N,\Rq)|$ many matrices $\T\in\Rq^{N\times m}$ such that $\rowspace(\T)=\Nspace$, and they are  obtained by fixing any matrix $\bar{\T}$ and considering $\{\A\bar{\T} \mid \A\in \GL(N,\Rq)\}$. Let us fix $\bar{\Mspace}:=\rowspace(\D\bar{\T})\in\mathrm{Mod}(\phi,\Nspace)$. We count for how many $A\in \GL(N,\Rq)$ we have $\rowspace(\D\A\bar{\T})=\bar{\Mspace}$. By Proposition \ref{prop:freemodule_matrix}, this happens if and only if there exist $\Y \in G_\phi^*, \Z \in H_\phi$ such that $\A\bar{\T}=\Y\bar{\T}+\Z$, which in turn is equivalent to the condition that there exists $\Y \in \G_\phi^*$ such that $(\A-\Y)\bar{\T}\in H_\phi$. Let us call $\S:=\A-\Y$ and divide $\S$ in $r\times r$ blocks $\S_{i,j}\in \Rq^{n_i\times n_j}$, for $i,j \in \{0,\ldots, r-1\}$. Divide also $\T$ in $r$ blocks $\T_i\in \Rq^{n_i\times m}$ for $i \in \{0,\ldots,r-1\}$. Hence, we have, for every $i \in \{0,\ldots,r-1\}$
 $$ \sum_{j=0}^{r-1} \S_{i,j}\T_j \in \maxIdeal^{r-i} \Rq^{n_i\times m}.$$
 Since the rows of $\T$ are linearly independent over $\Rq$, this implies that $\S_{i,j}\in \maxIdeal^{r-i}$, that is $\S$ is of the form
 \begin{equation*}
  \S=\A-\Y=\begin{bmatrix} 0 \\
   g_\maxIdeal^{r-1}\Z_1 \\ g_\maxIdeal^{r-2}\Z_2 \\ \vdots \\ g_\maxIdeal \Z_{r-1}\end{bmatrix}.
 \end{equation*}
 Therefore, we have $\rowspace(\D\A\bar{\T})=\bar{\Mspace}$ if and only if $\A=\Y+\S$. It is easy to see that this holds if and only if $\A\in G_{\phi}^*$. Hence, the $\Rq$-submodule $\bar{\Mspace}$ is counted $|G_{\phi}^*|$ many times. Since the choice of $\bar{\Mspace}$ was arbitrary, we conclude
 \begin{equation*}
   |\mathrm{Mod}(\phi,\Nspace)|=\frac{|\GL(N,\Rq)|}{|G_{\phi}^*|}. 
   \end{equation*}
\qed
\end{proof}
\end{appendices}

\end{document}